\documentclass{article}

\usepackage{arxiv}

\usepackage[utf8]{inputenc} 
\usepackage[T1]{fontenc}    
\usepackage{hyperref}       
\usepackage{url}            
\usepackage{booktabs}       
\usepackage{amsfonts}       
\usepackage{nicefrac}       
\usepackage{microtype}      
\usepackage{lipsum}

\usepackage{cite}
\usepackage{amsmath,amssymb,amsfonts}
\usepackage{graphicx}
\usepackage{algorithm}
\usepackage[noend]{algpseudocode}
\usepackage{hyperref}
\usepackage{textcomp}
\usepackage{threeparttablex} 

\usepackage{mathtools}
\usepackage{caption}
\usepackage{float}
\usepackage{stmaryrd}
\usepackage{epstopdf}
\usepackage{multirow}
\usepackage{pifont}
\usepackage{caption}
\usepackage{subcaption}

\newcommand{\norm}[1]{\left\lVert#1\right\rVert}
\newcommand{\R}{{\mathbb{R}}}

\newcommand{\y}{{y}}

\newcommand{\N}{{\mathbb{N}}}

\newcommand{\X}{{\mathbf{X}}}
\newcommand{\W}{{\mathbf{W}}}
\newcommand{\Y}{{\mathbf{Y}}}
\newcommand{\T}{{\mathbf{T}}}
\newcommand{\So}{{\mathbf{S}}}
\newcommand{\Obs}{{\mathcal{U}}}
\newcommand{\con}{{\mathcal{o}}}
\newcommand{\U}{{\mathbf{U}}}

\newcommand{\cmark}{\ding{51}}%
\newcommand{\xmark}{\ding{55}}%

\newtheorem{theorem}{Theorem}[section]

\newtheorem{assumption}{Assumption}

\newtheorem{definition}[theorem]{Definition}
\newtheorem{lemma}[theorem]{Lemma}
\newtheorem{remark}[theorem]{Remark}
\newtheorem{problem}[theorem]{Problem}
\newtheorem{proof}[theorem]{Proof}

\title{Spatiotemporal Tubes for Temporal Reach-Avoid-Stay Tasks in Unknown Systems
\thanks{ This work was supported in part by the SERB Start-Up Research Grant; in part by the ARTPARK. The work of Ratnangshu Das was supported by the Prime Minister’s Research Fellowship from the Ministry of Education, Government of India.}
}

\author{
 Ratnangshu Das\thanks{Authors contributed equally.} \\
  Robert Bosch Centre for Cyber-Physical Systems\\
  IISc, Bengaluru, India\\
  \texttt{ratnangshud@iisc.ac.in} \\
   \And
 Ahan Basu$^\dag$ \\
  Robert Bosch Centre for Cyber-Physical Systems\\
  IISc, Bengaluru, India\\
  \texttt{ahanbasu@iisc.ac.in} \\
  \And
 Pushpak Jagtap \\
  Robert Bosch Centre for Cyber-Physical Systems\\
  IISc, Bengaluru, India\\
  \texttt{pushpak@iisc.ac.in} \\
}

\begin{document}
\maketitle

\begin{abstract}
The paper considers the controller synthesis problem for general MIMO systems with unknown dynamics, aiming to fulfill the temporal reach-avoid-stay task, where the unsafe regions are time-dependent, and the target must be reached within a specified time frame. The primary aim of the paper is to construct the spatiotemporal tube (STT) using a sampling-based approach and thereby devise a closed-form approximation-free control strategy to ensure that system trajectory reaches the target set while avoiding time-dependent unsafe sets. The proposed scheme utilizes a novel method involving STTs to provide controllers that guarantee both system safety and reachability. In our sampling-based framework, we translate the requirements of STTs into a Robust optimization program (ROP). To address the infeasibility of ROP caused by infinite constraints, we utilize the sampling-based Scenario optimization program (SOP). Subsequently, we solve the SOP to generate the tube and closed-form controller for an unknown system, ensuring the temporal reach-avoid-stay specification. Finally, the effectiveness of the proposed approach is demonstrated through three case studies: an omnidirectional robot, a SCARA manipulator, and a magnetic levitation system.
\end{abstract}


\section{Introduction}
\label{sec:introduction}
Autonomous systems have been a focal point in control theory, owing to their broad range of applications, including safety-critical situations such as self-driving cars and unmanned aerial vehicles. A principal challenge in deploying these systems is reaching specific targets while avoiding unsafe regions and respecting state constraints. The study of these reach-avoid-stay (RAS) specifications \cite{Meng1} becomes even more crucial as they often form the foundation for more complex specifications \cite{Kloetzer, NAHS}. This underlines the importance of developing and implementing safe and reliable control strategies. 

The adoption of formal languages \cite{tabuda} to define tasks has led to the increasing popularity of symbolic control. This method abstracts the continuous state space into a finite symbolic model, simplifying the design of controllers with formal guarantees \cite{tabuda}. Enhancements such as the fixed-point algorithm for RAS controller synthesis \cite{FPA_RAS}, the integration of barrier certificates \cite{Local_Global} and incrementally building a directed tree to approximate the product automaton's state space \cite{Sampling_MAS} have improved upon the traditional abstraction-based method. However, as the system becomes more complex or the granularity of the abstraction increases, the symbolic model expands exponentially, leading to greater computational complexity.

Alternatively, control barrier functions (CBFs) \cite{CBF, jagtap2020formal} offer a discretization-free approach to controller synthesis and have been extensively used for safety-critical systems \cite{CBF_SC, CBF_GP}, particularly to meet obstacle avoidance \cite{C3BF}, a class of signal temporal logic (STL) and linear temporal logic (LTL) specifications (excluding avoid specification) \cite{CBF-STL, jagtap2020formal}, and RAS specifications \cite{Meng3}. While CBFs mitigate some of the scalability issues associated with symbolic control, they still rely on an optimization step, making them computationally demanding for high-dimensional systems. Additionally, CBFs require precise knowledge of system dynamics, which is a serious constraint for real-world systems. Another method for achieving RAS objectives under disturbances include Hamilton-Jacobi (HJ) reachability \cite{HJ_RAS}, which provides formal safety guarantees by solving a value function over the state space \cite{HJ_RAS}, but it scales poorly with system dimensionality \cite{HJE, HJ_review} and requires exact knowledge of system dynamics \cite{HJ_RAS2}. While learning-based techniques such as Gaussian Processes (GPs) and Neural Networks (NNs) can be used to estimate unknown dynamics \cite{DD_Safety_Filters}, they introduce their own limitations. GPs become computationally expensive with increasing dataset size, and NNs often lack formal guarantees.

Funnel-based control \cite{PPC1, surveyPPC} addresses these challenges by providing approximation-free closed-form control law that constrains tracking error within exponentially decaying funnel-shaped boundaries. This allows for rapid response to disturbances and changes in system dynamics without the overhead of real-time optimization. This approach has been used in tracking control problems in unknown nonlinear systems \cite{surveyPPC,mishra2023approximation} and multi-agent systems with intricate tasks \cite{Funnel_STL_MAS}. Despite its success in meeting reachability specifications \cite{NAHS, hard_soft}, handling specifications like avoiding unsafe areas remains significantly challenging \cite{Funnel_STL,lindemann2021funnel}. Some studies \cite{RB} combine path-planning algorithms with funnel-based tracking control to achieve RAS. However, separating the trajectory planning from funnel formulation can compromise avoid-constraints under disturbances. To address this \cite{KDF_1, KDF_2} enlarge the obstacles and construct paths in an extended free space. However, these works are restricted to static obstacles, limiting their applicability to general time-varying obstacles.

To solve this, the idea is to eliminate the trajectory generation step and directly design time-varying guidance functions in the state space, which ensures that the target set is reached within the prescribed time while avoiding time-dependent unsafe sets and adhering to state constraints. However, the main challenge lies in devising these time-varying guidance functions through the free state space. The spatiotemporal tube (STT) approach in \cite{STT, das2025control, das2025spatiotemporal, faruqui2025reach}, uses the circumvent function \cite{funnelarxiv} to define smooth time-varying tube functions. The tubes adapt around the unsafe regions, providing a safe channel for the trajectory to enforce RAS and prescribed-time RAS tasks. However, it assumes unsafe sets are unions of convex sets and handle only control-affine dynamics. Moreover, the abrupt STT adjustments, due to the circumvent function, significantly increase the control effort.

In this paper, we propose a sampling-based framework, inspired from \cite{DDC, MCBC, Lipschitz}, to design STTs that address temporal-reach avoid stay specifications (T-RAS), i.e., given an initial set, the system trajectory should reach a target set within desired time interval while avoiding any time-varying unsafe sets. Given a T-RAS specification, we first frame the conditions of STTs as a robust optimization problem (ROP). We proceed by sampling points in time and the unsafe set to establish a scenario optimization program (SOP) associated with the ROP. By solving the SOP, we construct the STTs that adhere to the T-RAS specification with a formal correctness guarantee. Finally, we derive a closed-form model-free control law that ensures the output of a general higher-order unknown MIMO system remains within these synthesized tubes, ensuring robust satisfaction of the T-RAS specification under bounded disturbances. We demonstrate the proposed approach on multiple case studies and compare it against HJ and CBF-based methods, highlighting its scalability, formal safety guarantees, and effectiveness in handling complex, high-dimensional systems.


\section{Preliminaries and Problem Formulation}
\subsection{Notation}

The symbols $\N$, $ \R$, $\R^+$, and $\R_0^+ $ denote the set of natural, real, positive real, and nonnegative real numbers, respectively. 
A vector space of real matrices with $ n $ rows and $ m $ columns is denoted by  $ \R^{n\times m} $. A column vector with $n$ rows is represented by $ \R^{n}$.
The Euclidean norm is represented using $\lVert \cdot \rVert$. For $a, b \in \mathbb{N}$ with $a \leq b$, the closed interval in $\mathbb{N}$ is denoted as $[a; b]$. 
A vector $x \in \mathbb{R}^{n}$ with entries $x_1, \ldots, x_n$ is represented as $[x_1, \ldots, x_n]^\top$, where $x_i \in \mathbb{R}$ denotes the $i$-th element of vector $x\in\mathbb{R}^n$ and $i \in [1;n]$.
A diagonal matrix in $\R^{n\times n}$ with diagonal entries $d_1,\ldots, d_n$ is denoted by $\textsf{diag}(d_1,\ldots, d_n)$.
Given a matrix $M\in\R^{n\times m}$, $M^\top$ represents transpose of matrix $M$. 
The power set of a set \textbf{A} is defined as $\mathcal{P}$(\textbf{A}).
{Given $N \in \N$ sets $\textbf{A}_i$, $i\in\left[1;N\right]$, we denote the Cartesian product of the sets by $\textbf{A}=\prod_{i\in\left[1;N\right]}\textbf{A}_i:=\{(x_1,\ldots,x_N)|x_i\in \textbf{A}_i,i\in\left[1;N\right]\}$.
The projection of a set $\textbf{A} \subset \mathbb{R}^n$ onto the $i$-th dimension, where $i \in [1; n]$, is represented by an interval $[\textbf{A}_{i,L}, \textbf{A}_{i,U}] \subset \mathbb{R}$, where:
$\textbf{A}_{i,L} := \min\{x_i \in \mathbb{R} \mid [x_1, \ldots, x_n] \in \textbf{A}\}$,
$\textbf{A}_{i,U} := \max\{x_i \in \mathbb{R} \mid [x_1, \ldots, x_n] \in \textbf{A}\}$.}

\subsection{System Definition}
Consider a class of control-affine MIMO nonlinear pure-feedback systems characterized by the following dynamics:
\begin{align} \label{eqn:sysdyn}
    &\dot{x}_i(t) = f_i(z_i(t)) + g_i(z_i(t))x_{i+1}(t) + w_i(t), i\in [1;N-1], \notag\\
    &\dot{x}_{N}(t) = f_N(z_N(t)) + g_N(z_N(t))u(t) + w_N(t), \\
    &y(t) = x_1(t), \nonumber
\end{align}
where for $t\in\R^+_0$ and $i\in[1;N]$,
\begin{itemize}
    \item $x_i(t) = [x_{i,1}(t), \ldots, x_{i,n}(t)]^\top \in {\X}_i \subset \mathbb{R}^{n}$ is the state vector,
    \item $z_i(t) := [x_1^\top(t),x_2^\top(t),...,x_i^\top(t)]^\top \in \overline{\X}_i = \prod_{j=1}^i \X_j \subset \mathbb{R}^{ni} $,
    \item $u(t) \in \mathbb{R}^n$ is control input vector,
    \item $w_i(t) \in \mathbf{W} \subset \R^n$ is unknown bounded external disturbance, and
    \item $y(t) = [x_{1,1}(t), \ldots, x_{1,n}(t)]\in \Y=\X_1$ denotes the output vector.
\end{itemize}


The functions $f_i: \overline{\X}_i \rightarrow \mathbb{R}^n$ , $g_i: \overline{\X}_i \rightarrow \mathbb{R}^{n \times n}, i \in [1;N]$, follows the Assumptions \ref{assum:lip} and \ref{assum:pd}.

\begin{assumption}\label{assum:lip}
    For all $i \in [1;N]$, functions $f_i$ and $g_i$ are unknown and locally Lipschitz.         
\end{assumption}
\begin{assumption}[\cite{PPC1,RAC}] \label{assum:pd}
    The matrix $g_{i,s}(z_i) := \frac{g_i(z_i)+g_i(z_i)^\top}{2}$ is uniformly sign definite with known signs for all $z_i \in \overline{\X}_i$. Without loss of generality, we assume $g_{i,s}(z_i)$ is positive definite, i.e., there exists a constant $\underline{g_i}\in\mathbb R^+, \forall i \in [1;N]$ such that
    $$0 < \underline{g_i} \leq \lambda_{\min} (g_{i,s}(z_i)), \forall \ z_i \in \overline{\X}_i,$$
    where $\lambda_{\min}(\cdot)$ represents the smallest eigenvalue of the matrix.
\end{assumption}

{This assumption ensures that in \eqref{eqn:sysdyn} global controllability is guaranteed, i.e., $g_{i,s}(z_i) \neq 0,$ for all $z_i \in \overline{\X}_i$.}

{
\begin{remark}
    While this paper considers a control-affine structure for the system in \eqref{eqn:sysdyn}, the results can be extended to nonaffine systems. By leveraging techniques from the literature, such as affine transformations \cite{affine_transform, affine_transform_taylor} through coordinate changes or Taylor series expansions, we can include the additional terms that represent the discrepancy between the actual nonaffine dynamics and the affine approximation as an unknown disturbance within the model. This extension would enable a broader application of the proposed approach beyond control-affine dynamics.
\end{remark}
}

\subsection{Problem Formulation}
Let the output of \eqref{eqn:sysdyn}, \textit{i.e.} $y(t)$, be subject to a \textit{temporal reach-avoid-stay specifications} defined next.
\begin{definition}[Temporal reach-avoid-stay (T-RAS) task]\label{def:tras}
Given an output-space $\mathbf{Y}=\X_{1}$, prescribed time $t_c \in \R^+$, a time-varying unsafe set $\U: \R_0^+ \rightarrow \mathcal{P}(\Y)$, an initial set $\So \subset \mathbf{Y} \setminus \U(0)$, and a target set $\T \subset \mathbf{Y} \setminus \U(t_c)$, we say that the output of the system satisfies temporal reach-avoid-stay specifications if $y(0) \in \So$, $y(t_c) \in \T$ and for all $s \in [0,t_c], y(s) \in \mathbf{Y} \setminus \U(s)$.
\end{definition}

\begin{remark}
    $\U(t) = \bigcup_{j=1}^p \Obs^{(j)}(t)$, where $\Obs^{(j)}(t)$ represents time-dependent connected unsafe set. Hence, $\U(t)$ can be disconnected, representing multiple time-varying obstacles. 
    Also, if $\mathbf{Y}$ has an arbitrary shape, we redefine the output space as $\hat{\Y} = \prod_{i=1}^{n} [\Y_{i,L}, \Y_{i,U}]$ and expand the unsafe set to $\hat{\U}(t) = \U(t) \cup (\hat{\mathbf{Y}} \setminus \mathbf{Y})$.
\end{remark}

\begin{problem} \label{problem:control}
Given the system in \eqref{eqn:sysdyn}, we aim to design an \textit{approximation-free}, \textit{closed-form} control law $u(t)$ ensuring the output $y(t)$ adheres to the T-RAS specification defined in Definition \ref{def:tras}.
\end{problem}



To solve the aforementioned problem, we leverage the STTs defined next.

\begin{definition}[Spatiotemporal Tubes for T-RAS task]\label{def:stt}
Given a T-RAS task in Definition \ref{def:tras}, time-varying intervals $[\gamma_{i,L}(t), \gamma_{i,U}(t)]$, where $\gamma_{i,L}:\R_0^+\rightarrow\R$ and $\gamma_{i,U}:\R_0^+\rightarrow\R$ are continuously differentiable functions with $\gamma_{i,L}(t) < \gamma_{i,U}(t)$, are called STTs for T-RAS, if for all $i \in [1;n]$, the following holds:
\begin{subequations}\label{eqn:stt}
\begin{align}
    &\prod_{i=1}^n [\gamma_{i,L}(t), \gamma_{i,U}(t)] \subseteq \Y, \forall t \in [0,t_c], \label{eqn:stta}\\
    &\prod_{i=1}^n [\gamma_{i,L}(0), \gamma_{i,U}(0)] \subseteq \So, \\
    &\prod_{i=1}^n [\gamma_{i,L}(t_c), \gamma_{i,U}(t_c)] \subseteq \T, \\
    &\prod_{i=1}^n 
    [\gamma_{i,L}(t), \gamma_{i,U}(t)] \cap \U(t) = \emptyset, \forall t \in [0,t_c].\label{eqn:sttd}
\end{align}
\end{subequations}
\end{definition}

\begin{remark}
    If one designs a control law that constrains the output trajectory within the STTs defined in Definition \ref{def:stt}, i.e.,
    \begin{align} \label{eqn:stt_con}
    \gamma_{i,L}(t) < y_i(t) < \gamma_{i,U}(t), \forall i \in [1;n], 
    \implies \gamma_L(t) < x_1(t) < \gamma_U(t),
\end{align}
then one can ensure the satisfaction of T-RAS specification.
\end{remark}

The STTs approach \cite{STT} increases the conservatism of the specifications by restricting the unsafe set to only a union of convex sets and requiring that the projection of the unsafe set does not overlap with the initial and target set in at least one dimension. {Additionally, the circumvent function—a time-varying function used to capture the unsafe set—introduces sharp changes in the tube shape (as shown in Figure 1 of \cite{STT}), which increases control effort.} To address these limitations and extend the concept of STTs of Definition \ref{def:stt} to more general T-RAS specifications mentioned in Definition \ref{def:tras}, we propose a sampling-based approach in this paper.




\section{Sampling-based Spatiotemporal Tubes Construction}\label{data_driven}

In this section, the main goal is to construct the STTs that start from the initial set and reach the target set, avoiding the obstacles denoted by the unsafe sets. 

In our sampling-based setting, we fix the structure of the curves that form the STTs for the $i$-th dimension as, 
$$\gamma_{i,\con}(c_{i,\con},t) = \sum_{k=1}^{z_{i,\con}} c_{i,\con}^k p_{i,\con}^k(t), \ \con \in {L,U}, \ i\in [1;n],$$
where $L$ and $U$ denote the lower and upper constraints, respectively. $p_{i,\con}(t)$ are user-defined nonlinear continuously differentiable basis functions and $c_{i,\con} = [c_{i,\con}^1; c_{i,\con}^2;...; c_{i,\con}^{z_{i,\con}}] \in \mathbb{R}^{z_{i,\con}}$ denote the unknown coefficients.

To satisfy the conditions in Definition \ref{def:stt}, we formulate the following Robust optimization program (ROP): 
\begin{subequations} \label{eq:ROP}
\begin{align}
& \min_{[d_1, d_2,...,d_n,\eta]} \quad \eta  \notag \\
& \textrm{s.t.} \notag \\
& \forall i= [1;n], \gamma_{i,\con}(c_{i,\con},0)\hspace{-0.2em} =\hspace{-0.2em} \hat{\So}_{i,\con}, \ \gamma_{i,\con}(c_{i,\con},t_c) = \hat{\T}_{i,\con};\\
& \forall (t,i) \in [0,t_c]\times[1;n]: \notag \\
& \quad \Y_{i,L} - \gamma_{i,L}(c_{i,L},t)  \leq \eta_i, 
\gamma_{i,U}(c_{i,U},t) - \Y_{i,U}  \leq \eta_i, \\
& \quad \gamma_{i,L}(c_{i,L},t) - \gamma_{i,U}(c_{i,U},t) + \gamma_{i,d} \leq \eta_i, \\
& \forall (t, \y) \in [0,t_c] \times \U(t), \exists i \in [1;n]: \notag \\
& \quad \min \left \{\y_i - \gamma_{i,L}(c_{i,L},t), \gamma_{i,U}(c_{i,U},t) - \y_i   \right\} \leq \eta_i; \\
& \forall i= [1;n], \eta_i \leq \eta, \\
& d_i = [c_{i,L}, c_{i,U}, \eta_i] \notag.
\end{align}
\end{subequations}
Here, $\con \in \{L,U\}$, $y(0) \in \prod_{i=1}^n [\hat{\So}_{i,L}, \hat{\So}_{i,U}] \subseteq \So$ and $y(t_c) \in \prod_{i=1}^n [\hat{\T}_{i,L}, \hat{\T}_{i,U}] \subseteq \T$. Also, $\gamma_{i,d} \in \R^+$ defines a minimum separation between the tubes.

One can readily observe that if the solution to the ROP $\eta^* \leq 0$, then it ensures the conditions in Definition \ref{def:stt} will be satisfied.


One faces two major challenges to solve the proposed ROP in (\ref{eq:ROP}). First, the ROP in (\ref{eq:ROP}) has infinitely many constraints since the output space and time are continuous. In addition, though knowledge of the unsafe set is priorly known, constructing the tube avoiding all points in the unsafe set leads to an infinite number of equations. 
To tackle these challenges, we aim to develop a sampling-based scheme for the construction of the curves that form the tube. 

To do so, we consider the augmented unsafe set $\W = [0,t_c] \times \U(t)$. Collect $N_t$ samples $w_r = (t_r, \y_r)$ from $\W$, where $r = [1;N_t]$. Consider a ball $\W_r $ around each sample $w_r$ with radius $\epsilon$, such that for all the points in the augmented space $(t,y) \in \W$, there exists a $w_r$ that satisfies 
\begin{align}\label{eq:ball}
    \lVert (t, \y) - w_r \rVert \leq \epsilon , \forall (t,\y) \in \W.
\end{align}  
This ensures that the union of all these balls forms a superset of the augmented space, i.e., $\bigcup_{r=1}^{N_t} \W_r \supset [0,t_c] \times \U(t)$. 

Now, we construct the Scenario optimization program (SOP) associated with the ROP:
\begin{subequations} \label{eq:SOP}
\begin{align}
& \min_{[d_1, d_2,...,d_n,\eta]} \quad \eta \notag \\
& \textrm{s.t.}  \notag \\
& \forall i \in [1;n]: \gamma_{i,\con}(c_{i,\con},0)\hspace{-0.2em} =\hspace{-0.2em} \hat{\So}_{i,\con}, \ \gamma_{i,\con}(c_{i,\con},t_c) = \hat{\T}_{i,\con}; \\
& \forall r \in [1;N_t], (t_r,i) \in [0,t_c]\times[1;n]: \notag \\
& \quad \Y_{i,L} - \gamma_{i,L}(c_{i,L},t_r)  \leq \eta_i,  
\gamma_{i,U}(c_{i,U},t_r) - \Y_{i,U}  \leq \eta_i,\\
& \quad \gamma_{i,L}(c_{i,L},t_r) - \gamma_{i,U}(c_{i,U},t_r) + \gamma_{i,d} \leq \eta_i, \\
& \forall r \in [1;N_t], (t_r, \y_r) \in \W, \exists i \in [1;n]: \notag \\
& \quad \min \left\{\y_{i,r} - \gamma_{i,L}(c_{i,L},t_r), \gamma_{i,U}(c_{i,U},t_r) - \y_{i,r}\right\} \leq \eta_i; \\
& \forall i= [1;n], \eta_i \leq \eta, \\
& d_i = [c_{i,L}, c_{i,U},\eta_i] \notag.
\end{align}
\end{subequations}

Here $\y_{r} = [\y_{1,r}, \ldots, \y_{n,r}]^\top$. One can readily observe that SOP in \eqref{eq:SOP} has a finite number of constraints of the same form as (\ref{eq:ROP}).

Now, to guarantee that the tubes formed by solving the SOP in \eqref{eq:SOP}, fulfill the constraints of ROP in \eqref{eq:ROP}, we assume the following: 

\begin{assumption} \label{assum:funlip}
    $\gamma_{i,L}(c_{i,L},t)$ and $\gamma_{i,U}(c_{i,U},t)$ are Lipschitz continuous with respect to $t$ with Lipschitz constants $\mathcal{L}_{L}$ and $\mathcal{L}_{U}$ for all $i \in [1;n]$.
\end{assumption}

Under Assumption \ref{assum:funlip}, Theorem \ref{th:constr} outlines a sampling-based methodology for constructing STTs with a certified confidence of 1.

\begin{theorem} \label{th:constr}
    Under Assumption \ref{assum:funlip}, suppose the SOP in (\ref{eq:SOP}) is solved with $N_t$ sampled data as in (\ref{eq:ball}). Let the optimal value of SOP be $\eta_S^*$ with solution $d_i^* = [c_{i,L}^*, c_{i,U}^*, \eta_i^*]$ $\forall i \in [1;n]$. If 
    \begin{equation} \label{eq:satisfy}
        \eta_S^* + \mathcal{L}\epsilon \leq 0,
    \end{equation}
    where $\mathcal{L} = \max\{ \mathcal{L}_{L}, \mathcal{L}_{U}, \mathcal{L}_{L} + \mathcal{L}_{U}, \mathcal{L}_{L}+1, \mathcal{L}_{U}+1\}$, then the STTs functions $\gamma_{i,L}\text{ and } \gamma_{i,U}$, $\forall i \in [1;n]$, obtained from the SOP in \eqref{eq:SOP} ensures the satisfaction of conditions of Definition \ref{def:stt}.
\end{theorem}

\begin{proof}\label{proof:constr}
    First, we demonstrate that under condition \eqref{eq:satisfy}, the $\gamma_{i,L}(t)$ and $\gamma_{i,U}(t)$ constructed through solving the SOP in \eqref{eq:SOP} satisfy Equation \eqref{eqn:stta}. The optimal $\eta_S^*$, obtained through solving the \eqref{eq:SOP}, guarantees for any $r \in [1;N_t], (t_r,i) \in [0,t_c] \times [1;n],$ we have: 
    \begin{align*}
    &\Y_{i,L} - \gamma_{i,L}(c_{i,L},t_r) \leq \eta_S^*,\\
    &\gamma_{i,U}(c_{i,L},t_r) - \Y_{i,U} \leq \eta_S^*, \\
    &\gamma_{i,L}(c_{i,L}, t_r) - \gamma_{i,U}(c_{i,U},t_r) +\gamma_{i,d} \leq \eta_S^*,
    \end{align*}
    Now from \eqref{eq:ball}, we infer that $\forall t \in [0,t_c], \exists \hspace{0.2em} t_r $ s.t. $|t-t_r| \leq \epsilon$.
    Thus, $\forall i \in [1;n], \forall r \in [1;N_t]$, $\forall t\in [0,t_c]$:
    \begin{itemize}
        \item[(a)] $\Y_{i,L} - \gamma_{i,L}(c_{i,L},t) =\Y_{i,L} - \gamma_{i,L}(c_{i,L},t_r) + \gamma_{i,L}(c_{i,L},t_r) -\gamma_{i,L}(c_{i,L},t)\leq \eta_S^* + \mathcal{L}_{L} |t-t_r| \leq \mathcal{L}\epsilon + \eta_S^* \leq 0$,\\
        \item[(b)] $\gamma_{i,U}(c_{i,U},t) - \Y_{i,U}=\gamma_{i,U}(c_{i,U},t) - \gamma_{i,U}(c_{i,U},t_r) +\gamma_{i,U}(c_{i,U},t_r) - \Y_{i,U}\leq \mathcal{L}_{U} |t-t_r| + \eta_S^* \leq \mathcal{L}\epsilon + \eta_S^* \leq 0$,\\
        \item[(c)] $\gamma_{i,L}(c_{i,L},t) - \gamma_{i,U}(c_{i,U},t) + \gamma_{i,d} \\=\left(\gamma_{i,L}(c_{i,L},t) - \gamma_{i,L}(c_{i,L},t_r)\right) + \big(\gamma_{i,L}(c_{i,L},t_r) -\gamma_{i,U}(c_{i,U},t_r) + \gamma_{i,d}\big) + \left(\gamma_{i,U}(c_{i,U},t_r) - \gamma_{i,U}(c_{i,U},t)\right)\\\leq \mathcal{L}_{L} |t-t_r| +\eta_S^* + \mathcal{L}_{U} |t-t_r| \leq (\mathcal{L}_{L} + \mathcal{L}_{U})\epsilon + \eta_S^* \\ \leq \mathcal{L}\epsilon + \eta_S^* \leq 0$.
    \end{itemize}
    
    Next, we show that when the condition in \eqref{eq:satisfy} is satisfied, the $\gamma_{i,L}(c_{i,L},t)$ and $\gamma_{i,U}(c_{i,U},t)$ formulated by solving the SOP in \eqref{eq:SOP} satisfy argument \eqref{eqn:sttd}.
    The optimal $\eta_i^*$, obtained from solving the SOP in \eqref{eq:SOP}, also ensures that for any $r \in [1;N_t], (t_r, \y_r) \in \W$, 
    $\min \left\{\y_{i,r} - \gamma_{i,L}(c_{i,L},t_r), \gamma_{i,U}(c_{i,U},t_r) - \y_{i,r}\right\} \leq \eta_i^*.$ 
     Now, $\forall i \in [1;n], \forall r \in [1;N_t]$, $\forall t\in [0,t_c]$, $\forall \y_i\in \U(t)$:
     \begin{align*}
         & \y_i - \gamma_{i,L}(c_{i,L},t) \\ 
         & = \left(\y_i - \y_{i,r}\right) + \left(\y_{i,r} - \gamma_{i,L}(c_{i,L},t_r)\right) + \big(\gamma_{i,L}(c_{i,L},t_r) - \gamma_{i,L}(c_{i,L},t)\big) \leq \epsilon + \eta_i^* +  \mathcal{L}_{L}|t-t_r| \\
         & \leq (\mathcal{L}_{L} + 1)\epsilon + \eta_S^* \leq \mathcal{L} \epsilon + \eta_S^* \leq 0\\
         &\hspace{-1em}\text{Or}\\
         &\gamma_{i,U}(c_{i,U},t) - \y_i \\
         & =  \left(\y_{i,r} - \y_i\right) + \left(\gamma_{i,U}(c_{i,U},t_r) - \y_{i,r}\right) + \big(\gamma_{i,U}(c_{i,U},t) - \gamma_{i,U}(c_{i,U},t_r)\big) \leq \epsilon + \eta_i^* +  \mathcal{L}_{U}|t-t_r| \\
         & \leq (\mathcal{L}_{U} + 1)\epsilon + \eta_S^* \leq \mathcal{L} \epsilon + \eta_S^* \leq 0.
     \end{align*}
     Therefore, if the condition in \eqref{eq:satisfy} is met, the STT constructed with boundaries defined by $\gamma_{i,L}(c_{i,L},t)$ and $\gamma_{i,U}(c_{i,U},t)$, for all $i \in [1;n]$ as determined by solving the SOP in \eqref{eq:SOP} satisfies Definition \ref{def:stt}, thereby completing the proof.
\end{proof}
\begin{remark}
    Note that the Lipschitz constants $\mathcal{L}_L$ and $\mathcal{L}_U$ are required to check condition \eqref{eq:satisfy} in Theorem \ref{th:constr}. We introduce Algorithm \ref{algo:Lipschitz} to estimate these Lipschitz constants in Appendix A, which follows a similar procedure as \cite[Algorithm 1]{MCBC} and \cite[Algorithm 2]{Lipschitz}. 
\end{remark}

\begin{remark}
    The basis functions $p_{i,\con}^k$ that shape the STTs are highly flexible and can be chosen in various forms, such as monomials in $t$ for a polynomial-type STT. Unlike CBFs, which depend on the multi-dimensional state variable $x$ \cite{MCBC}, STTs rely solely on $t$, significantly reducing computational complexity. Additionally, data-driven CBF synthesis requires sampling from the system's dynamics, whereas the sampling-based STT approach only requires data samples from time and the unsafe set, simplifying the data collection process.
    
    {The computational complexity of the SOP depends on factors like the choice of basis functions and the sampling density. Notably, in complex environments, with a greater number of unsafe sets, higher-degree polynomials provide greater flexibility but also increase the number of decision variables, leading to a polynomial growth in computation time.  Similarly, finer sampling in time and unsafe regions adds more constraints to the SOP, further contributing to a polynomial increase in computational demand. The change in computational complexity with the number of decision variables and constraints is illustrated in Figure \ref{fig:compute}.
    However, it is important to note that the tube computation is performed offline and does not affect the real-time implementation of the STT, which relies on a closed-form control law, discussed in the next section.}
\end{remark}
\begin{figure}[h!]
     \centering
     \begin{subfigure}[b]{0.4\textwidth}
         \centering
         \includegraphics[width=\textwidth]{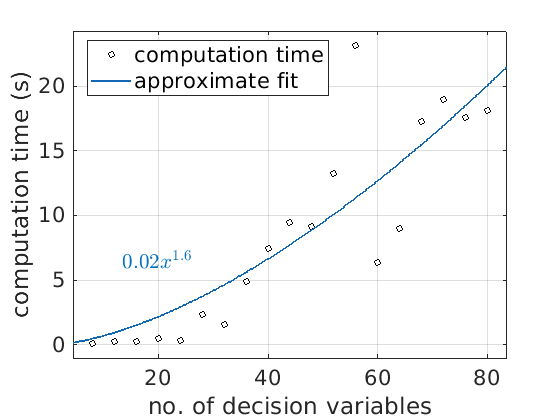}
         \caption{Computation Time vs Decision Variables}
     \end{subfigure}
     \hfill
     \begin{subfigure}[b]{0.4\textwidth}
         \centering
         \includegraphics[width=\textwidth]{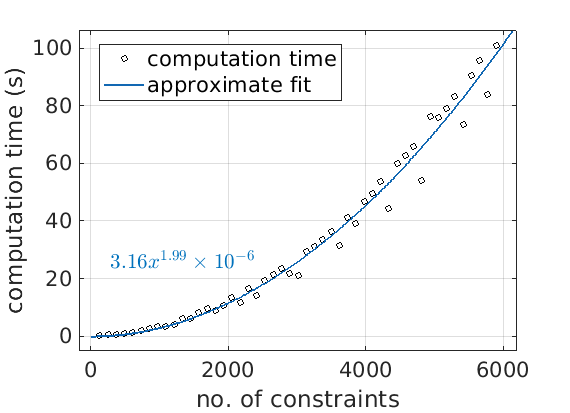}
         \caption{Computation Time vs Constraints}
     \end{subfigure}
        \caption{Time complexity}
        \label{fig:compute}
\end{figure}

\section{Controller Design}\label{sec:control}
In this section, we use the STTs developed through the sampling-based approach outlined in \eqref{eq:SOP} to derive an approximation-free, closed-form control law to constrain the output of system \eqref{eqn:sysdyn} within the tubes. The lower triangular structure of \eqref{eqn:sysdyn} allows us to use a backstepping-like design approach similar to that described in \cite{feedback}. First, we design an intermediate control input $r_2$ for the $x_1$ dynamics to ensure the fulfillment of \eqref{eqn:stt_con}. We then iteratively design the intermediate control laws $r_{k+1}$ for the $x_k$ dynamics, ensuring $x_k$ tracks $r_k$, for all $k \in [2;N]$. It is important to note that $r_{N+1}$ effectively becomes the actual control input $u$ for the system. The steps of the controller design are outlined below.

\textbf{Stage $1$}: Given $\gamma_{1,i,L}(c_{i,L} ,t)$ and $\gamma_{1,i,U}(c_{i,U} ,t)$, $i\in[1;n]$, define the normalized error $e_1(x_1,t)$, the transformed error $\varepsilon_1(x_1,t)$ and the diagonal matrix $\xi_1(x_1,t)$ as
\begin{subequations} \label{eq:stage I}
   \begin{align}
    e_1(x_1,t) &= [e_{1,1}(x_{1,1},t), \ldots, e_{1,n}(x_{1,n},t)]^\top = (\gamma_{1,d} (t))^{-1} \left( 2x_1 - \gamma_{1,s} (t) \right), \\
    \varepsilon_1(x_1,t) &= \big[\ln\left(\frac{1+e_{1,1}(x_{1,1},t)}{1-e_{1,1}(x_{1,1},t)}\right), \ldots, \ln\left(\frac{1+e_{1,n}(x_{1,n},t)}{1-e_{1,n}(x_{1,n},t)}\right) \big]^\top ,\\
    \xi_1(x_1,t) &= \frac{4 (\gamma_{1,d} (t))^{-1}}{1-e_1^\top(x_1,t)e_1(x_1,t)},
    \end{align} 
\end{subequations}
where, $\gamma_{1,s} := [\gamma_{1,1,U} + \gamma_{1,1,L}, \ldots, \gamma_{1,n,U} + \gamma_{1,n,L}]^\top$ and $\gamma_{1,d} := \textsf{diag}(\gamma_{1,1,d},\ldots,\gamma_{1,n,d})$, with $\gamma_{1,i,d} = \gamma_{1,i,U} - \gamma_{1,i,L}$.

The intermediate control input $r_2(x_1,t)$ is given by: 
\begin{equation*}
    r_2(x_1,t) = - \kappa_1\varepsilon_1(x_1,t)\xi_1(x_1,t), \kappa_1 \in \R^+.
\end{equation*}

\textbf{Stage $k$} ($k \in [2;N]$): {Given the reference vector $r_k(z_k,t)$, we aim to design the subsequent intermediate control $r_{k+1}(z_{k+1},t)$ for the dynamics of $x_{k+1}$, ensuring that $x_{k+1}$ tracks the trajectory determined by $r_k(z_k,t)$. Here, $z_k = [x_1,x_2,\ldots, x_{k-1}]^\top$.}

This is done by enforcing exponentially narrowing constraint $\gamma_{k,i}(t) = (p_{k,i} - q_{k,i})e^{-\mu_{k,i}t} + q_{k,i}$, such that 
\begin{align*}
    -\gamma_{k,i}(t) \leq (x_{k,i}-r_{k,i}) \leq \gamma_{k,i}(t) \ \ \forall (t,i) \in \R_0^+ \times [1;n].
\end{align*}
{Here, $\mu_{k,i} \in \R_0^+$, and $p_{k,i}, q_{k,i} \in \R^+$ with $p_{k,i} > q_{k,i}$ and $|x_{k,i}(t=0) - r_{k,i}(t=0)| \leq p_{k,i}$.}

Now define the normalized error $e_k(x_{k},t)$, the transformed error $\varepsilon_k(x_{k},t)$ and the diagonal matrix $\xi_k(x_{k},t)$ as
\begin{subequations} \label{eq:stage k}
    \begin{align}
    e_k(x_{k},t) &= [e_{k,1}(x_{k,1},t), \ldots, e_{k,n}(x_{k,n},t)]^\top = (\gamma_{k,d} (t))^{-1} \left(x_{k} - r_k \right), \\
    \varepsilon_k(x_{k},t) &= \big[\ln\left(\frac{1+e_{k,1}(x_{k,1},t)}{1-e_{k,1}(x_{k,1},t)}\right), \ldots, \ln\left(\frac{1+e_{k,n}(x_{k,n},t)}{1-e_{k,n}(x_{k,n},t)}\right) \big]^\top, \\
    \xi_k(x_{k},t) &= \frac{4 (\gamma_{k,d} (t))^{-1}}{1-e_k^\top(x_{k},t)e_k(x_{k},t)},
\end{align}
\end{subequations}
where $\gamma_{k,d} := \textsf{diag}(\gamma_{k,1,d},\ldots,\gamma_{k,n,d})$, with $\gamma_{k,i,d} = \frac{1}{2} (\gamma_{k,i,U} - \gamma_{k,i,L}) = \gamma_{k,i} \nonumber, \forall i \in [1;n]$.

So, the intermediate control inputs $r_{k+1}(z_{k+1},t)$ to enforce the desired temporal reach-avoid-stay task is given by 
\begin{equation*}
    r_{k+1}(z_{k+1},t) = - \kappa_k\varepsilon_k(x_{k},t)\xi_k(x_{k},t), \kappa_k \in \R^+.
\end{equation*}
Note that, at the $N$-th stage, $r_{N+1}(z_{N+1},t)$ essentially serves as the actual control input $u$, which is given by,
\begin{equation*}
    u(z_{N+1},t) = - \kappa_N\varepsilon_N(x_{N},t)\xi_N(x_{N},t), \kappa_N \in \R^+.
\end{equation*}

Thus, we can design the control input to perform the T-RAS specification for the system described in \eqref{eqn:sysdyn}.

\begin{theorem} \label{theorem_ras}
    Given a nonlinear MIMO system in \eqref{eqn:sysdyn} satisfying assumptions \ref{assum:lip} and \ref{assum:pd}, a temporal reach-avoid-stay (T-RAS) task as defined in Definition \ref{def:tras}, and a spatio-temporal tube as discussed in Section \ref{data_driven}, if the output $y(0)$ is within the STTs at time $t=0$, i.e., $\gamma_{1,i,L}(0) < y_{i}(0) < \gamma_{1,i,U}(0), \forall i \in [1;n]$, then the closed-form control strategies,\begin{subequations}\label{eqn:Control_ras}
     \begin{align}
        r_{k+1}(z_{k+1},t) &= - \kappa_k\varepsilon_k(x_{k},t)\xi_k(x_{k},t), k \in [1;N-1], \\
        u(z_{N+1},t) &= - \kappa_N\varepsilon_N(x_{N},t)\xi_N(x_{N},t),
    \end{align}    
    \end{subequations}
    will ensure the satisfaction of given temporal reach-avoid-stay task where $z_{k+1} = [x_1,x_2,\ldots, x_k]^\top$ and $\varepsilon_k$, $\xi_k$ are shown in \eqref{eq:stage I} and \eqref{eq:stage k}.
\end{theorem}

\begin{proof}
The proof is done for the stages mentioned above. First, we prove the control law for the STT, and then we sequentially prove it for the other stages. 

In each stage, the proof proceeds in three steps.  First, we show that there exists a maximal solution for the normalized error $e_k: [0,\tau_{\max}] \rightarrow \mathbb{D}$, which implies that $e_k(x_{k},t)$ remains within $\mathbb{D}$ in the maximal time solution interval $[0, \tau_{\max})$. Next, we show that the proposed control law in \eqref{eqn:Control_ras} constraints $e_k(x_{k},t)$ to a compact subset of $\mathbb{D}$. Finally, we prove that $\tau_{\max}$ can be extended to $\infty$.\\

\textbf{Stage $1$:}
Differentiating the normalized error $e_1(x_1,t)$ w.r.t time and substituting the first equation of the system dynamics \eqref{eqn:sysdyn} we get,
\begin{align*}
     \dot{e}_1 =& 2(\gamma_{1,d}(t))^{-1} \Big( f_1(x_1) + g_1(x_1)x_2 + w_1  - \frac{1}{2}(\dot{\gamma}_{1,s}(t) + \dot{\gamma}_{1,d}(t) e_1)\Big) :=h_1(t,e_1),
\end{align*}
where $x_1 = \frac{\gamma_{1,d}(t)e_1 + \gamma_{1,s}(t)}{2}$.
We also define the constraints for $e_1$ through the open and bounded set $\mathbb{D}:=(-1, 1)^n$. \\

\textit{\underline{Step $(i)$}}:
Since the initial state $x_1(0)$ satisfies $\gamma_{1,i,L}(0) < x_{1,i}(0) < \gamma_{1,i,U}(0), \forall i \in [1;n]$, the initial normalized error $e_1(0)$ is also within the constrained region $\mathbb{D}$. Further, the STTs are bounded and continuously differentiable functions of time, the functions $f_1(x_1)$ and $g_1(x_1)$ are locally Lipschitz and the control law $r_2(x_1,t)$ is smooth over $\mathbb{D}$. As a consequence, $h_1(t,e_1)$ is bounded and continuously differentiable on $t$ and locally Lipschitz on $e_1$ over $\mathbb{D}$. 

Therefore, according to \cite[Theorem 54]{sontag}, there exists a maximal solution to the initial value problem $\dot{e}_1 = h_1(t,e_1), e_1(0) \in \mathbb{D}$ on the time interval $[0, \tau_{\max})$ such that $e_1(t) \in \mathbb{D}, \forall t \in [0, \tau_{\max}).$\\

\textit{\underline{Step $(ii)$}}:
Consider the following positive definite and radially unbounded Lyapunov function candidate: $V_1 = \frac{1}{2}\varepsilon_1^\top\varepsilon_1$. 

Differentiating $V_1$ with respect to $t$ and substituting $\dot{\varepsilon_1}$, $\dot{e_1}$, system dynamics \eqref{eqn:sysdyn}, and the control strategy (\ref{eqn:Control_ras}), we have:
\begin{align*}
    \dot{V}_1 &= \varepsilon_1^T \dot{\varepsilon}_1= \varepsilon_1^T \frac{2}{1-e_1^Te_1}\dot{e}_1 = \varepsilon_1^T \xi_1 \left(\dot{x}_1 - \frac{1}{2}(\dot{\gamma}_{1,s}+\dot{\gamma}_{1,d}e_1)\right) \\
    &= \varepsilon_1^T \xi_1 \left( f_1(x_1) + g_1(x_1)x_2 +w_1  - \frac{1}{2}(\dot{\gamma}_{1,s}+\dot{\gamma}_{1,d}e_1) \right), \\
    &= \varepsilon_1^T \xi_1 \left( f_1(x_1) -\kappa_1 g_1(x_1) \xi_1 \varepsilon_1 + w_1 -\frac{1}{2}(\dot{\gamma}_{1,s}+\dot{\gamma}_{1,d} e_1) \right).
\end{align*}
Using Rayleigh-Ritz inequality and Assumption \ref{assum:pd},
\begin{align*}
    \underline{g_1}\|\varepsilon_1\|^2 \|\xi_1\|^2 \leq \lambda_{\min}(g_1(x_1))\|\varepsilon_1\|^2 \|\xi_1\|^2  \leq \varepsilon_1^\top \xi_1 g_1(x_1) \xi_1 \varepsilon_1, \\
    -\kappa_1 \varepsilon_1^\top \xi_1 g_1(x_1) \xi_1 \varepsilon_1 \leq -\kappa_1 \underline{g_1}\|\varepsilon_1\|^2 \|\xi_1\|^2 = - \kappa_g^1\|\varepsilon_1\|^2 \|\xi_1\|^2.
\end{align*}
Therefore, 
$\dot{V}_1 \leq -\kappa_g^1\|\varepsilon_1\|^2 \|\xi_1\|^2 + \|\varepsilon_1\| \|\xi_1\| \|\Phi_1\|,$

where $\Phi_1 := f_1(x_1) + w_1 - \frac{1}{2}\dot{\gamma}_{1,s} - \frac{1}{2}\dot{\gamma}_{1,d}e_1$. 
From the construction of $\gamma_{1,s}, \gamma_{1,d}$  we know that $\dot{\gamma}_{1,s}$ and $\dot{\gamma}_{1,d}$ are bounded by construction. From step 1, we have $e_1(t) \in \mathbb{D}$ and consequently $\forall i \in [1;n]$, $x_{1,i}(t) \in (\gamma_{1,i,L}(t), \gamma_{1,i,U}(t))$. Thus, owing to the continuity of $f_1(x_1)$ and employing the extreme value theorem, we can infer $\|f_1(x_1)\| < \infty$. 
Hence, $\|\Phi_1\| < \infty, \forall t \in [0, \tau_{max})$. 

Now add and substract $\kappa_g^1\theta\norm{\varepsilon_1}^2 \|\xi_1\|^2$, where $\theta\in(0,1)$.
\begin{align*}
    &\dot{V}_1 \leq -\kappa_g^1(1-\theta)\norm{\varepsilon_1}^2 \|\xi_1\|^2 - \norm{\varepsilon_1}\|\xi_1\| \left(\kappa_g^1 \theta \norm{\varepsilon_1}\|\xi_1\| - \|\Phi_1\| \right) \\
    &\leq -\kappa_g^1(1-\theta)\norm{\varepsilon_1}^2 \norm{\xi_1}^2, \forall \kappa_g^1 \theta \norm{\varepsilon_1}\|\xi_1\| - \|\Phi_1\| \geq 0 \\
    &\leq -\kappa_g^1(1-\theta)\norm{\varepsilon_1}^2 \|\xi_1\|^2, \forall \norm{\varepsilon_1} \geq \frac{\|\Phi_1\|}{\kappa_g^1 \theta \|\xi_1\|} \nonumber,\forall t \in [0, \tau_{\max}).
\end{align*}
Therefore, we can conclude that there exists a time-independent upper bound $\varepsilon_1^* \in \mathbb{R}_{0}^+$ to the transformed error $\varepsilon_1$, i.e., $\|\varepsilon_1\| \leq \varepsilon_1^*, \forall t \in [0, \tau_{\max})$.

Consequently, taking inverse logarithmic function,
\begin{align*}
    -1 < \frac{e_{1,i}^{-\varepsilon_{1,i}^*}-1}{e_{1,i}^{-\varepsilon_{1,i}^*}+1} =: e_{1,i,L} \leq e_{1,i} \leq e_{1,i,U} := \frac{e_{1,i}^{\varepsilon_{1,i}^*}-1}{e_{1,i}^{\varepsilon_{1,i}^*}+1} < 1,     \forall t \in [0, \tau_{\max}), \ \text{for } i \in [1;n].
\end{align*}
Therefore, by employing the control law (\ref{eqn:Control_ras}), we can constrain $e_1$ to a compact subset of $\mathbb{D}$ as:
\begin{align} \label{eqn:e_compact}
    e_1(t) \in [e_{1,L}, e_{1,U}] =: \mathbb{D}' \subset \mathbb{D}, \forall t \in [0, \tau_{\max}), \nonumber
\end{align}
where $e_{1,L} = [e_{1,1,L}, \ldots, e_{1,n,L}]^\top$ and $e_{1,U} = [e_{1,1,U}, \ldots, e_{1,n,U}]^\top$.\\

\textit{\underline{Step $(iii)$}}:
Finally, we prove $\tau_{\max}$ can be extended to $\infty$. 
We know that $e_1(t) \in \mathbb{D}', \forall t \in [0, \tau_{\max})$, where $\mathbb{D}'$ is a non-empty compact subset of $\mathbb{D}$.
However, if $\tau_{\max} < \infty$ then according to \cite[Proposition C.3.6]{sontag}, $\exists t' \in [0, \tau_{\max})$ such that $e_1(t) \notin \mathbb{D}$. This leads to a contradiction!
Hence, we conclude that $\tau_{\max}$ can be extended to $\infty$.\\

\textbf{Stage $k$ ($k\in [2;N]$):}

Differentiating the normalized error $e_k(x_{k},t)$ with respect to time and substituting the corresponding equations of the system dynamics \eqref{eqn:sysdyn}, we get
\begin{align*}
    \dot{e_k} = (\gamma_{k,d}(t))^{-1} ( f_k(x_{k}) + g_k(x_{k})x_{k+1} + w_k  - (\dot{r}_k(z_k,t) + \dot{\gamma}_{k,d}(t) e_k)) : =h_k(t,e_k).
\end{align*}
We also define the constraints for $e_k$ through the open and bounded set $\mathbb{D}:=(-1, 1)^n$. \\

\textit{\underline{Step $(i)$}}:
Since the initial state $x_{k}(0)$ satisfies $-\gamma_{k,i}(0) < x_{k,i}(0) < \gamma_{k,i}(0) ,\forall i \in [1;n]$, the initial normalized error $e_k(0)$ is also within the constrained region $\mathbb{D}$. Further, the STTs are bounded and continuously differentiable functions of time, the functions $f_k(x_{k})$ and $g_k(x_{k})$ are locally Lipschitz and the control law $r_{k+1}(z_{k+1},t)$ is smooth over $\mathbb{D}$. As a consequence, $h_k(t,e_k)$ is bounded and continuously differentiable on $t$ and locally Lipschitz on $e_k$ over $\mathbb{D}$. 

Therefore, according to \cite[Theorem 54]{sontag}, there exists a maximal solution to the initial value problem $\dot{e_k} = h_k(t,e_k), e_k(0) \in \mathbb{D}$ on the time interval $[0, \tau_{\max})$ such that $e_k(t) \in \mathbb{D}, \forall t \in [0, \tau_{\max}).$\\

\textit{\underline{Step $(ii)$}}:
Consider the following positive definite and radially unbounded Lyapunov function candidate: $V_k = \frac{1}{2}\varepsilon_k^\top\varepsilon_k$. 

Differentiating $V_k$ with respect to $t$ and substituting $\dot{\varepsilon_k}$, $\dot{e_k}$, system dynamics \eqref{eqn:sysdyn}, and the control strategy (\ref{eqn:Control_ras}), we obtain:
\begin{align*}
    &\dot{V}_k = \varepsilon_k^T \dot{\varepsilon}_k= \varepsilon_k^T \frac{2}{1-e_k^Te_k}\dot{e}_k \\
    & = \frac{1}{2} \varepsilon_k^T \xi_k \left(\dot{x}_k - \dot{r}_k (z_k,t) - \dot{\gamma}_{k,d}e_k\right) \\
    &= \frac{1}{2} \varepsilon_k^T \xi_k \big( f_k(x_{k}) + g_k(x_{k})x_{k+1} +w_k  - \dot{r}_k(z_k,t) - \dot{\gamma}_{k,d}e_k\big) \\
    &= \frac{1}{2} \varepsilon_k^T \xi_k \big( f_k(x_{k}) -\kappa_k g_k(x_{k}) \xi_k \varepsilon_k + w_k - \dot{r}_k(z_k,t) - \dot{\gamma}_{k,d}e_k\big).
\end{align*}
Using Rayleigh-Ritz inequality and Assumption \ref{assum:pd},
\begin{align*}
    \underline{g_k}\|\varepsilon_k\|^2 \|\xi_k\|^2 & \leq \lambda_{\min}(g_k(x_{k}))\|\varepsilon_k\|^2 \|\xi_k\|^2 \leq \varepsilon_k^\top \xi_k g_k(x_{k}) \xi_k \varepsilon_k, \\
    \implies -\frac{1}{2}\kappa_k \varepsilon_k^\top \xi_k g_k(x_{k}) \xi_k \varepsilon_k & \leq -\frac{1}{2} \kappa_k \underline{g_k}\|\varepsilon_k\|^2 \|\xi_k\|^2 = - \kappa_g^k\|\varepsilon_k\|^2 \|\xi_k\|^2.
\end{align*}
Therefore, 
$\dot{V}_k \leq -\kappa_g^k\|\varepsilon_k\|^2 \|\xi_k\|^2 + \|\varepsilon_k\| \|\xi_k\| \|\Phi_k\|,$
where $\Phi_k := \frac{1}{2}\left(f_k(x_{k}) + w_k - \dot{r}_k(z_k,t) - \dot{\gamma}_{k,d}e_k\right)$. 
From the construction of $ \gamma_{k,d}$ we know that $\dot{\gamma}_{k,d}$ is bounded . From step k-a, we have $e_k(t) \in \mathbb{D}$ and consequently $\forall i \in [1;n]$, $x_{k,i}(t) \in (-\gamma_{k,i}(t), \gamma_{k,i}(t))$. Thus, owing to the continuity of $f_k(x_{k})$ and employing the extreme value theorem, we can infer $\|f_k(x_{k})\| < \infty$. Also, since $r_k$ is bounded, we can say that $\dot{r}_k$ is bounded.  
Hence, $\|\Phi_k\| < \infty, \forall t \in [0, \tau_{max})$. 

Now add and substract $\kappa_g^k\theta\norm{\varepsilon_k}^2 \|\xi_k\|^2$, where $\theta\in(0,1)$.
\begin{align*}
    \dot{V_k} &\leq -\kappa_g^k(1-\theta)\norm{\varepsilon_k}^2 \|\xi_k\|^2 - \norm{\varepsilon_k}\|\xi_k\| \left(\kappa_g^k \theta \norm{\varepsilon_k}\|\xi_k\| - \|\Phi_k\| \right) \\
    &\leq -\kappa_g^k(1-\theta)\norm{\varepsilon_k}^2 \norm{\xi_k}^2, \forall \kappa_g^k \theta \norm{\varepsilon_k}\|\xi_k\| - \|\Phi_k\| \geq 0 \\
    &\leq -\kappa_g^k(1-\theta)\norm{\varepsilon_k}^2 \|\xi_k\|^2, \forall \norm{\varepsilon_k} \geq \frac{\|\Phi_k\|}{\kappa_g^k \theta \|\xi_k\|}, \ \forall t \in [0, \tau_{\max}).
\end{align*}
Therefore, we can conclude that there exists a time-independent upper bound $\varepsilon_k^* \in \mathbb{R}_{0}^+$ to the transformed error $\varepsilon_k$, i.e., $\|\varepsilon_k\| \leq \varepsilon_k^*, \forall t \in [0, \tau_{\max})$.

Consequently, taking inverse logarithmic function,
\begin{align*}
    -1 < \frac{e_{k,i}^{-\varepsilon_{k,i}^*}-1}{e_{k,i}^{-\varepsilon_{k,i}^*}+1} =: e_{k,i,L} \leq e_{k,i} \leq e_{k,i,U} := \frac{e_{k,i}^{\varepsilon_{k,i}^*}-1}{e_{k,i}^{\varepsilon_{k,i}^*}+1} < 1, \    \forall t \in [0, \tau_{\max}), \ \text{for } i \in [1;n].
\end{align*}

Therefore, by employing the control law (\ref{eqn:Control_ras}), we can constrain $e_k$ to a compact subset of $\mathbb{D}$ as:
\begin{align*}
    e_k(t) \in [e_{k,L}, e_{k,U}] =: \mathbb{D}' \subset \mathbb{D}, \forall t \in [0, \tau_{\max}), \nonumber
\end{align*}
where, $e_{k,L} = [e_{k,1,L}, \ldots, e_{k,n,L}]^\top$ and $e_{k,U} = [e_{k,1,U}, \ldots, e_{k,n,U}]^\top$.\\

\textit{\underline{Step $(iii)$}}:
Finally, we prove $\tau_{\max}$ can be extended to $\infty$. 
We know that $e_k(t) \in \mathbb{D}', \forall t \in [0, \tau_{\max})$, where $\mathbb{D}'$ is a non-empty compact subset of $\mathbb{D}$.
However, if $\tau_{\max} < \infty$ then according to \cite[Proposition C.3.6]{sontag}, $\exists t' \in [0, \tau_{\max})$ such that $e_k(t) \notin \mathbb{D}$. This leads to a contradiction!
Hence, we conclude that $\tau_{\max}$ can be extended to $\infty$.


Thus, the control strategy in \eqref{eqn:Control_ras} solves the T-RAS task as mentioned in Definition \ref{def:tras}, hence completing the proof.
\end{proof}
\begin{remark}
Note that the closed-form time-dependent control law \eqref{eqn:Control_ras} is approximation-free and guarantees the satisfaction of T-RAS specifications for control affine MIMO pure-feedback systems with unknown dynamics.  
\end{remark}


\section{Case Study}

To show the effectiveness of our approach to constructing the STT-based controllers for the temporal reach-avoid-stay tasks, we consider three different case studies: (i) an omnidirectional robot, (ii) a 2-link manipulator, and (iii) a magnetic levitator system. 

\subsection{Omnidirectional Robot}

Consider the omnidirectional robot adopted from \cite{NAHS} and defined as 
\begin{align}
    \begin{bmatrix}
        \dot{x}_1 \\ \dot{x}_2 \\ \dot{x}_3
    \end{bmatrix}
    = 
    \begin{bmatrix}
        \cos{x_3} & -\sin{x_3} & 0 \\ \sin{x_3} & \cos{x_3} & 0 \\ 0 & 0 & 1
    \end{bmatrix}
    \begin{bmatrix}
        v_1 \\ v_2 \\ \omega
    \end{bmatrix} + w(t),
\end{align}
where the state vector $[x_1, x_2, x_3]^\top$ captures the robot's pose, $[v_1, v_2, \omega]^\top$ is the input velocity vector in the robot's frame, and $w$ is an external disturbance, with output of the system being $ y = [x_1,x_2]^\top$.

The robot is operating in a 2D environment with obstacles in the presence of unknown bounded external disturbance. The unsafe set covers the start zone (\textit{i.e.,} initial state) and goal (\textit{i.e.,} target state) in such a way that the adaptive tube proposed in \cite{STT} can not be framed from the initial state to the target set, avoiding the unsafe region. 

The starting zone of the robot is $\So = [1,1.5]\times[2,2.5] $ while the target region is $\T = [4.5,5]\times [4.5,5] $. The time bound for reaching the target is considered as $t_c = 5$. The unsafe set (red colored), as shown in Figure \ref{fig:sim1}(a), is assumed to present at the same location for the complete time-horizon $[0,t_c]$. As discussed in Section \ref{data_driven}, we consider the template for the STTs as second-order polynomials in time $t$: {$p_{i,\con}(t) = \{1, t, t^2\}$, for $\con = \{L,U\}$}. The STT obtained to solve the T-RAS task for this case is given by the following curves:
\begin{align*}
    \gamma_{1,L}(c_{1,L},t) &= 1 + 0.2377t + 0.0925t^2,\\
    \gamma_{1,U}(c_{1,U},t) &= 1.5 - 0.0023t + 0.1405t^2,\\
    \gamma_{2,L}(c_{2,L},t) &= 2 - 1.9782t + 0.4956t^2,\\
    \gamma_{2,U}(c_{2,U},t) &= 2.5 - 2.2183t + 0.5437t^2.
\end{align*}

The Lipschitz constants $\mathcal{L}_L$ and $\mathcal{L}_U$ are estimated to be $2.93$ and $3.17$, respectively, so $\mathcal{L} = 6.1$ and $\epsilon$ is considered as $0.0005$. We have used the Z3 SMT solver \cite{z3} to solve the scenario optimization problem \eqref{eq:SOP}. The $\eta_S^*$ upon solving the SOP is $-0.1$, which essentially means $\eta_S^* + \mathcal{L}\epsilon = -0.1 + 6.1\times0.0005 = -0.09695< 0$, that follows Theorem \ref{th:constr}. The computation time to compute the STTs via solving SOP in \eqref{eq:SOP} is 6.58 seconds with 2000 sample points.

The trajectory of the robot under the influence of the proposed control law in \eqref{eqn:Control_ras} as well as the obtained STTs in each dimension are shown in Figure \ref{fig:sim1}(a). 

We consider another case where the starting zone of the robot is $\So = [0,0.5]\times[0,0.5] $ while the target is located at the zone $\T = [4.5,5]\times [4.5,5] $. The time required to reach the target $t_c = 10$. The unsafe set, as shown in Figure \ref{fig:sim1}(b), is present in the same location for the complete time horizon $[0,t_c]$. Here, we consider the template for the STTs as third-order polynomials in time $t$: {$p_{i,\con}(t) = \{1, t, t^2, t^3\}$, for $\con = \{L,U\}$}. The STTs obtained to solve the T-RAS task for this case are given by the following curves:
\begin{align*}
    \gamma_{1,L}(c_{1,L},t) &= 3.9463t - 0.9857t^2 + 0.0636t^3,\\
    \gamma_{1,U}(c_{1,U},t) &= 0.5 + 3.8711t - 0.9928t^2 + 0.0651t^3,\\
    \gamma_{2,L}(c_{2,L},t) &= 0.4283t - 0.0009t^2 + 0.0001t^3,\\
    \gamma_{2,U}(c_{2,U},t) &= 0.5 + 0.1945t + 0.0422t^2 - 0.0017t^3.
\end{align*}

The trajectory of the robot under the influence of the proposed control law in \eqref{eqn:Control_ras} as well as the obtained STTs in each dimension are shown in Figure \ref{fig:sim1}(b).

\begin{figure*}[t]
    \centering
    \includegraphics[width =\textwidth]{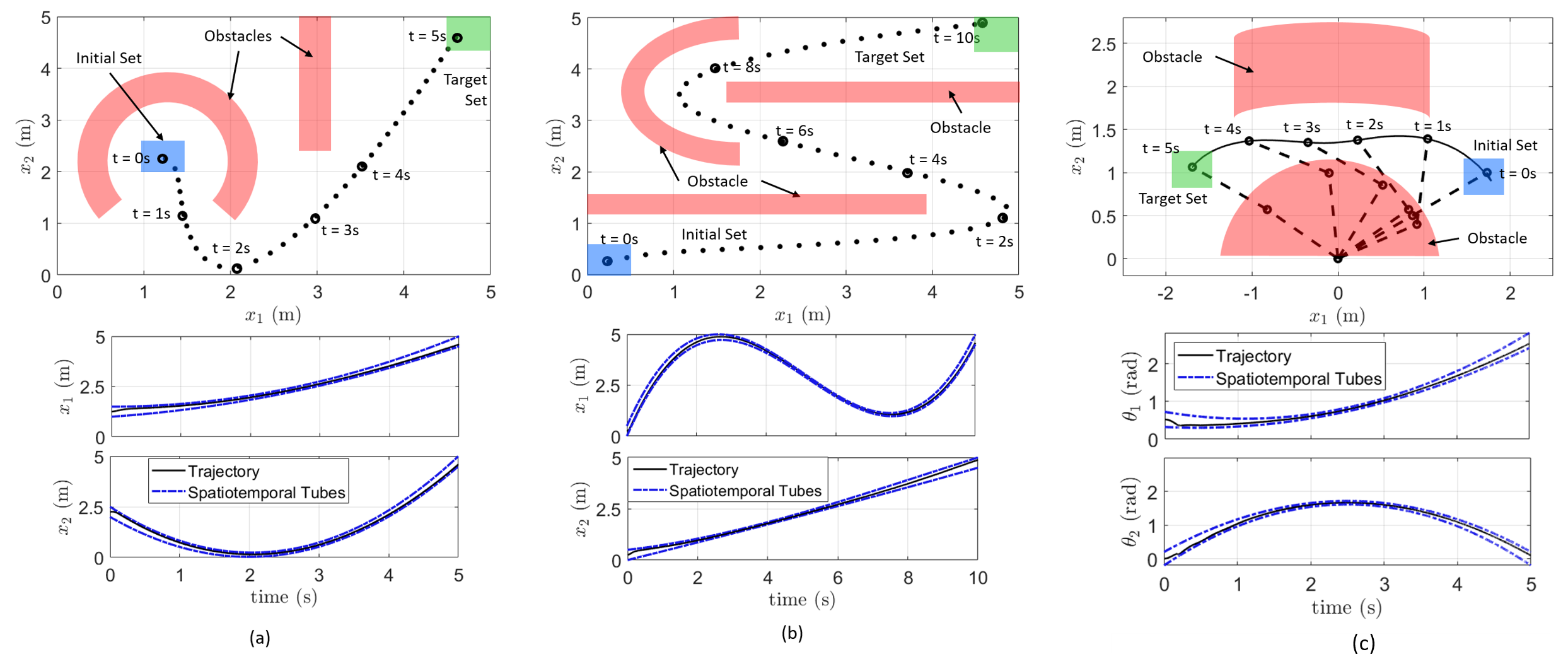}
    \caption{(a) (b) Trajectory of omnidirectional robot navigating 2D environments (top) and respective STTs (bottom) (c) Trajectory of 2R manipulator (top) and respective STTs (bottom).}
    \label{fig:sim1}
\end{figure*}

\begin{figure*}[t]
    \centering
    \includegraphics[width = \textwidth]{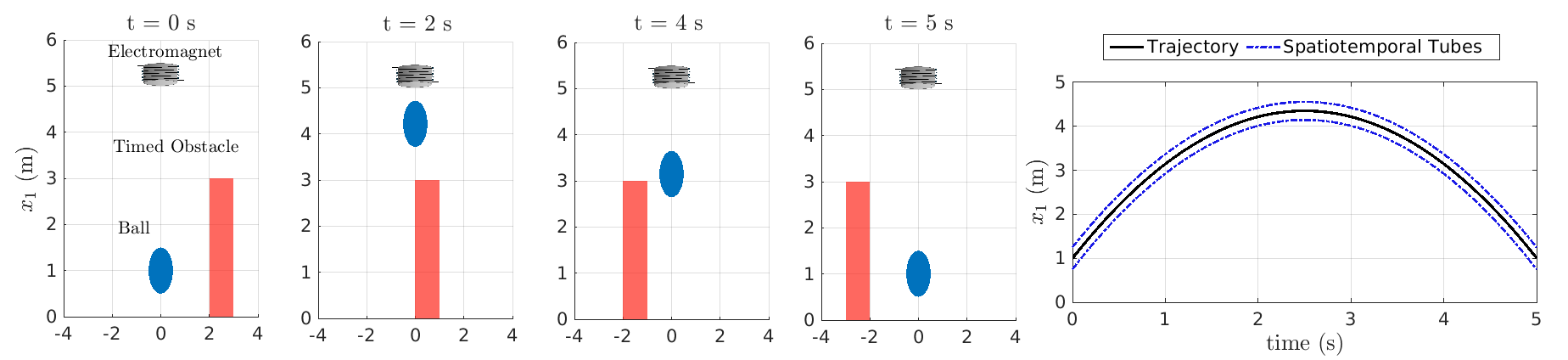}
    \caption{Temporal evolution of magnetic levitation system using the control law in \eqref{eqn:Control_ras} (left) and respective STTs for T-RAS tasks (right).}
    \label{fig:sim3}
\end{figure*}

\begin{table}[h]
\centering
\resizebox{\textwidth}{!}{
\begin{threeparttable}
\caption{Comparing STTs with classical algorithms}
\begin{tabular}{lcccccccccccc}
\hline
\textbf{Algorithm} & \multicolumn{2}{c}{\textbf{\begin{tabular}[c]{@{}l@{}}Closed-form \\ Control \end{tabular}}} & \multicolumn{2}{c}{\textbf{\begin{tabular}[c]{@{}l@{}}Formal \\ Guarantee \end{tabular}}} & \multicolumn{2}{c}{\textbf{\begin{tabular}[c]{@{}l@{}}Prescribed-time \\ Reachability \end{tabular}}} & \multicolumn{2}{c}{\textbf{\begin{tabular}[c]{@{}l@{}}Unknown \\ Dynamics \end{tabular}}} & \multicolumn{2}{c}{\textbf{\begin{tabular}[c]{@{}l@{}}Bounded \\ Disturbance \end{tabular}}} & \multicolumn{2}{c}{\textbf{\begin{tabular}[c]{@{}l@{}}Time dependent \\ Obstacle \end{tabular}}} \\ \hline
RRT*\cite{RRTs} & \multicolumn{2}{c}{-\tnote{1}} & \multicolumn{2}{c}{\xmark} & \multicolumn{2}{c}{-\tnote{1}} & \multicolumn{2}{c}{\cmark} & \multicolumn{2}{c}{-\tnote{1}} & \multicolumn{2}{c}{\cmark} \\
MPC\cite{MPC1}\cite{MPC2} & \multicolumn{2}{c}{\cmark} & \multicolumn{2}{c}{\cmark} & \multicolumn{2}{c}{\xmark} & \multicolumn{2}{c}{\xmark} & \multicolumn{2}{c}{\xmark} & \multicolumn{2}{c}{\cmark} \\ 
CBF-based methods \cite{C3BF}\cite{Meng3} & \multicolumn{2}{c}{\xmark} & \multicolumn{2}{c}{\cmark} & \multicolumn{2}{c}{\xmark} & \multicolumn{2}{c}{\xmark} & \multicolumn{2}{c}{\xmark} & \multicolumn{2}{c}{\cmark} \\
{HJ-based reachability}\cite{HJ_review} & \multicolumn{2}{c}{\xmark} & \multicolumn{2}{c}{\cmark} & \multicolumn{2}{c}{\xmark} & \multicolumn{2}{c}{\xmark} & \multicolumn{2}{c}{\cmark} & \multicolumn{2}{c}{\cmark} \\
{NN control}\cite{NN_con_book} & \multicolumn{2}{c}{\xmark} & \multicolumn{2}{c}{\xmark} & \multicolumn{2}{c}{\cmark} & \multicolumn{2}{c}{\cmark} & \multicolumn{2}{c}{\xmark} & \multicolumn{2}{c}{\xmark} \\
Symbolic Control\cite{SCOTS} & \multicolumn{2}{c}{\xmark} & \multicolumn{2}{c}{\cmark} & \multicolumn{2}{c}{\xmark} & \multicolumn{2}{c}{\xmark} & \multicolumn{2}{c}{\xmark} & \multicolumn{2}{c}{\xmark} \\
STT & \multicolumn{2}{c}{\cmark} & \multicolumn{2}{c}{\cmark} & \multicolumn{2}{c}{\cmark} & \multicolumn{2}{c}{\cmark} & \multicolumn{2}{c}{\cmark} & \multicolumn{2}{c}{\cmark} \\
\hline
\end{tabular}\label{tab:comp}
\begin{tablenotes}
    \item[1] Additional mechanisms, are required to ensure control, satisfy reachability within the prescribed time, and handle bounded disturbances. 
\end{tablenotes}
\end{threeparttable}
}
\end{table}

\subsection{Euler-Lagrange System}
For a second example, we consider a two-link 2R SCARA manipulator adapted from \cite{spong}. The system operates via two rotating joints, with the joint angles represented by $\theta_1$ and $\theta_2$. We define the system's output at any time as $y(t) = [\theta_1(t), \theta_2(t)]^\top$. Using the STT approach, we ensure that the system trajectory starting from the initial output $\So = [\frac{\pi}{6}, 0]^\top$, reaches the target output $\T = [\frac{5\pi}{6}, 0]^\top$, bypassing the unsafe sets shown in Figure \ref{fig:sim1}(c). The unsafe set is present for the complete time-horizon $[0,t_c]$ with $t_c = 5$. Note that to design the STTs, the obstacle in the workspace has to be mapped to an equivalent unsafe set in the state space. The following describes the model used:
\begin{gather}
    ml^2
    \begin{bmatrix}
        \frac{5}{3} + c_2 & \frac{1}{3} + \frac{1}{2}c_2 \\
        \frac{1}{3} + \frac{1}{2}c_2 & \frac{1}{3}
    \end{bmatrix}
    \begin{bmatrix}
        \Ddot{\theta}_1 \\
        \Ddot{\theta}_2
    \end{bmatrix} 
    +
    ml^2s_2
    \begin{bmatrix}
        -\frac{1}{2}\dot{\theta}_2^2 - \dot{\theta}_1\dot{\theta}_2\\
        \frac{1}{2}\dot{\theta}_2^2
    \end{bmatrix}
    + mgl
    \begin{bmatrix}
        \frac{3}{2}c_1 + \frac{1}{2}c_{12} \\
        \frac{1}{2}c_{12}
    \end{bmatrix}
    = 
    \begin{bmatrix}
        \tau_1(t) \\
        \tau_2(t)
    \end{bmatrix}
    + d(t), \nonumber
\end{gather}
where $m = 1 kg$ and $l = 1 m$ are the mass and length of each link, $g = 9.8m/s^2$ is the acceleration due to gravity, $\tau_1(t), \tau_2(t)$ are the torque inputs at the joints, $d(t)$ is an unknown bounded disturbance, $c_1 = \cos{\theta_1}$, $c_2 = \cos{\theta_2}$, $s_2 = \sin{\theta_2}$, and $c_{12} = \cos{(\theta_1 + \theta_2)}$.

Here, we consider the template for the STTs as second-order polynomials in time $t$: {$p_{i,\con}(t) = \{1, t, t^2\}$, for $\con = \{L,U\}$}. The STTs obtained to solve the T-RAS task for this case are given by the following curves:
\begin{align*}
    \gamma_{1,L}(c_{1,L},t) &= 0.3236 - 0.0893t + 0.1016t^2, \\
    \gamma_{1,U}(c_{1,U},t) &= 0.7236 - 0.3293t + 0.1496t^2,\\
    \gamma_{2,L}(c_{2,L},t) &= -0.2002 + 1.4496t - 0.2899t^2,\\
    \gamma_{2,U}(c_{2,U},t) &= 0.2000 + 1.2097t - 0.2419t^2.
\end{align*}

The Lipschitz constants $\mathcal{L}_L$ and $\mathcal{L}_U$ are estimated to be $1.408$ and $1.215$, respectively, so $\mathcal{L} = 2.623$ and for $\epsilon = 0.00002$, the $\eta_S^*$ upon solving the SOP is $-0.0001$, which essentially means $\eta_S^* + \mathcal{L}\epsilon = -0.0001 + 2.623\times0.00002 = -0.0000475 < 0$, that follows Theorem \ref{th:constr}. The computation time to compute the STTs is 206 seconds with 40500 sample points. 

The simulation results with the proposed control law in \eqref{eqn:Control_ras}, as well as the obtained STTs in each dimension, are depicted in Figure \ref{fig:sim1}(c). 

\subsection{Magnetic Levitator System}
Consider the magnetic levitator system adopted from \cite{Inc_stab} and defined as:
\begin{align*}
    \dot{x}_1 &= \frac{x_2}{M},
    \dot{x}_2 = \frac{x_3}{2\alpha} - Mg, 
    \dot{x}_3 = -\frac{2R}{\alpha}(1-x_1)x_3 + 2\sqrt{x_3}u,
\end{align*}
where the states $x_1,x_2,x_3$ denote the ball's position, the momentum of the ball and the square of the flux linkage associated with the electromagnet, respectively with output $ y = x_1$. Also, $M=1$ represents the mass of the ball, $g=9.8$ stands for acceleration due to gravity, $R=10$ denotes coil resistance around the electromagnet, $\alpha=0.5$ is a positive constant that depends on the number of coil turns, and $u$ represents the voltage applied to the electromagnet. The system is expected to avoid timed obstacles, keeping initial and target regions the same. The start and target region are considered to be the same and are $\So = \T = [0.75,1.25]$. To simulate the time-dependent obstacle, we assume that the obstacle moves from right to left and it will have a possible collision during the time interval [1.5,3.5]. The time-dependent unsafe set is defined by
\[
\U(t) =
\begin{cases}
[0,3], & \text{if } t \in [1.5,3.5], \\
\phi, & \text{elsewhere, } 
\end{cases}
\] where $\phi$ represents empty set. 

Here, we consider the template for the STTs as second-order polynomials in time $t$: {$p_{i,\con}(t) = \{1, t, t^2\}$, for $\con = \{L,U\}$}. The STT obtained to solve the T-RAS task for this case is given by the following curves:
\begin{align*}
    \gamma_{1,L}(c_{1,L},t) = 0.75 + 2.7167t - 0.5433t^2,\\
    \gamma_{1,U}(c_{1,U},t) = 1.25 + 2.6447t - 0.5289t^2.
\end{align*}
The simulation result with the proposed control law in \eqref{eqn:Control_ras}, as well as the obtained STT, are shown in Figure \ref{fig:sim3}.

\subsection{Drones}
We consider the state of the drone, adapted from \cite{APF}, is given by the following equation:
\begin{align*}
    [\dot{x}_1, \dot{x}_2, \dot{x}_3]^\top &= [v_x, v_y, v_z]^\top + w_1(t), \\
    [\dot{v}_x, \dot{v}_y, \dot{v}_z]^\top &= [u_x, u_y, u_z]^\top + w_2(t),
\end{align*}
where $[x_1,x_2,x_3]^\top$ captures the position of the drone, $[v_x,v_y,v_z]^\top$ is the velocity of the drone and $[u_x,u_y,u_z]^\top$ is the control input of the drone.

We consider the drone to be diagonally moving in an arena given by $[0,3]\times[0,3]\times[0,15]$, starting from $\So = [2.75,3]\times [2.75,3] \times [0,0.25]$ reaching towards the target $\T = [0,0.25] \times [0,0.25] \times [0,0.25]$, avoiding the static unsafe set given by $\U = [1,2]\times[0,3]\times[0,3]$ which is assumed to be present there for the complete time-horizon. The time bound for reaching the target is considered as $t_c = 20$.   We consider another dynamic cubic obstacle with a uniform width of $0.25 m$ along each dimension. This obstacle moves along the opposite diagonal of the arena, with the trajectory of its center specified by the following parametric equations:
\begin{gather*}
    x_{1o}(t) = 2.875-0.1375t, \
    x_{2o}(t) = 0.125+0.1375t, \
    x_{3o}(t) = -0.1t^2+2t+0.125.
\end{gather*}
Here, we consider the template for the STTs as second-order polynomials in time $t$: {$p_{i,\con}(t) = \{1, t, t^2\}$, for $\con = \{L,U\}$}. The STT obtained to solve the T-RAS task for this case is given by the following curves:
\begin{align*}
    \gamma_{1,L}(c_{1,L},t) &= 2.75 - 0.0296t - 0.0054t^2,\\
    \gamma_{1,U}(c_{1,U},t) &= 3 - 0.0396t - 0.0049t^2, \\
    \gamma_{2,L}(c_{2,L},t) &= 2.75 - 0.1336t - 0.0002t^2,\\
    \gamma_{2,U}(c_{2,U},t) &= 3 - 0.1436t + 0.0003t^2, \\
    \gamma_{3,L}(c_{3,L},t) &= 0 + 1.9175t - 0.0959t^2,\\
    \gamma_{3,U}(c_{3,U},t) &= 0.25 + 1.9075t - 0.0954t^2.
\end{align*}
The trajectory of the drones under the influence of the proposed control law in \eqref{eqn:Control_ras}, as well as the obtained sampling-based STTs for the drone in each dimension is shown in Figure \ref{fig:drone}. 

\begin{figure*}[h!]
     \centering
     \begin{subfigure}[b]{0.45\textwidth}
         \centering
         \includegraphics[width=0.8\textwidth]{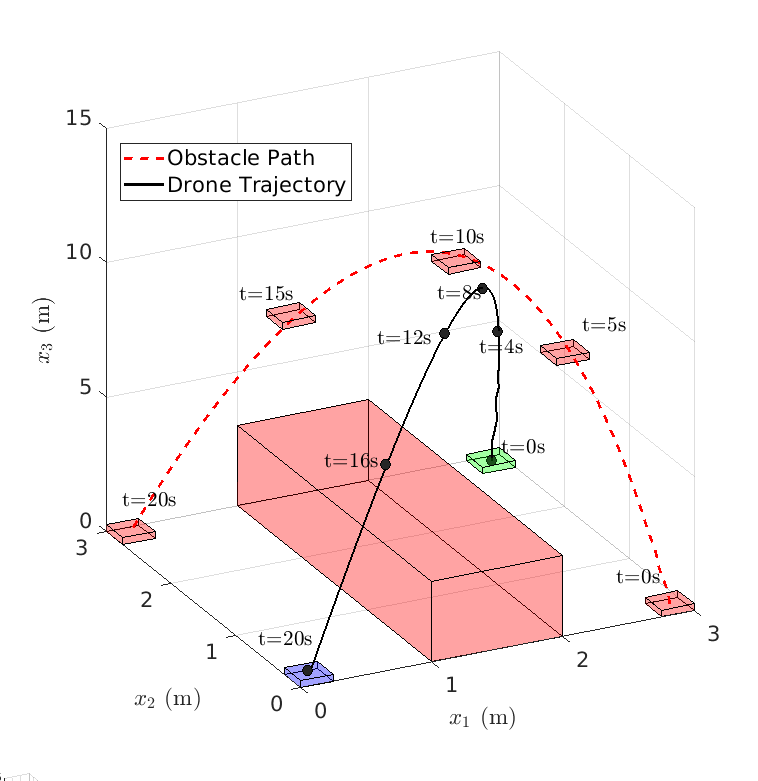}
         \caption{}
     \end{subfigure}
     \begin{subfigure}[b]{0.45\textwidth}
         \centering
         \includegraphics[width=0.8\textwidth]{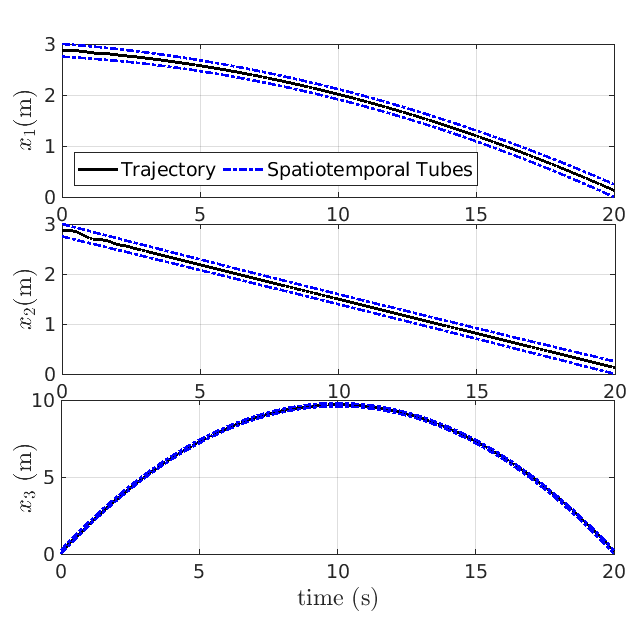}
         \caption{}
     \end{subfigure}
        \caption{Trajectory of drone with static and dynamic obstacles in a 3D environment (left) and respective STTs (right)}
        \label{fig:drone}
\end{figure*}

All the computations were performed on a machine with a Linux Ubuntu operating system with Intel i7-7700 CPU and 32GB RAM. 

\begin{table}[h!]
\centering
\caption{Quantitative Comparison of STT with classical algorithms}
\resizebox{0.72\linewidth}{!}{%
\begin{tabular}{|c|c|c|c|}
\hline
 &  & \begin{tabular}[c]{@{}c@{}}Offline Computation Time (s)\end{tabular} & \begin{tabular}[c]{@{}c@{}}Online Control Synthesis Time (s)\end{tabular} \\ \hline
\multirow{3}{*}{\begin{tabular}[c]{@{}c@{}}Single Integrator \\ (Reach-Avoid)\end{tabular}} & \textbf{STT} & \textbf{13.992} & \textbf{0.021} \\ \cline{2-4} 
 & HJB & 124.289 & 24.696 \\ \cline{2-4} 
 & CBF & - & 3.429 \\ \hline
\multirow{3}{*}{\begin{tabular}[c]{@{}c@{}}Double Integrator \\ (Reach-Avoid)\end{tabular}} & \textbf{STT} & \textbf{13.992} & \textbf{0.029} \\ \cline{2-4} 
 & HJB & 211.792 & 276.355 \\ \cline{2-4} 
 & CBF & - & 5.667 \\ \hline
\multirow{3}{*}{\begin{tabular}[c]{@{}c@{}}1R manipulator \\ (Reach)\end{tabular}} & \textbf{STT} & \textbf{0.233} & \textbf{0.036} \\ \cline{2-4} 
 & HJB & 0.702 & 1.929 \\ \cline{2-4} 
 & CBF & - & 0.537 \\ \hline
\multirow{3}{*}{\begin{tabular}[c]{@{}c@{}}2R Manipulator \\ (Reach)\end{tabular}} & \textbf{STT} & \textbf{0.234} & \textbf{0.062} \\ \cline{2-4} 
 & HJB & 25916.048 & 891.992 \\ \cline{2-4} 
 & CBF & - & 2.367 \\ \hline
\end{tabular}\label{tab:comp2}
}
\end{table}

\begin{figure}[h!]
\centering
\resizebox{0.75\linewidth}{!}{%
\begin{tabular}{cc}
    \includegraphics[width=0.45\textwidth]{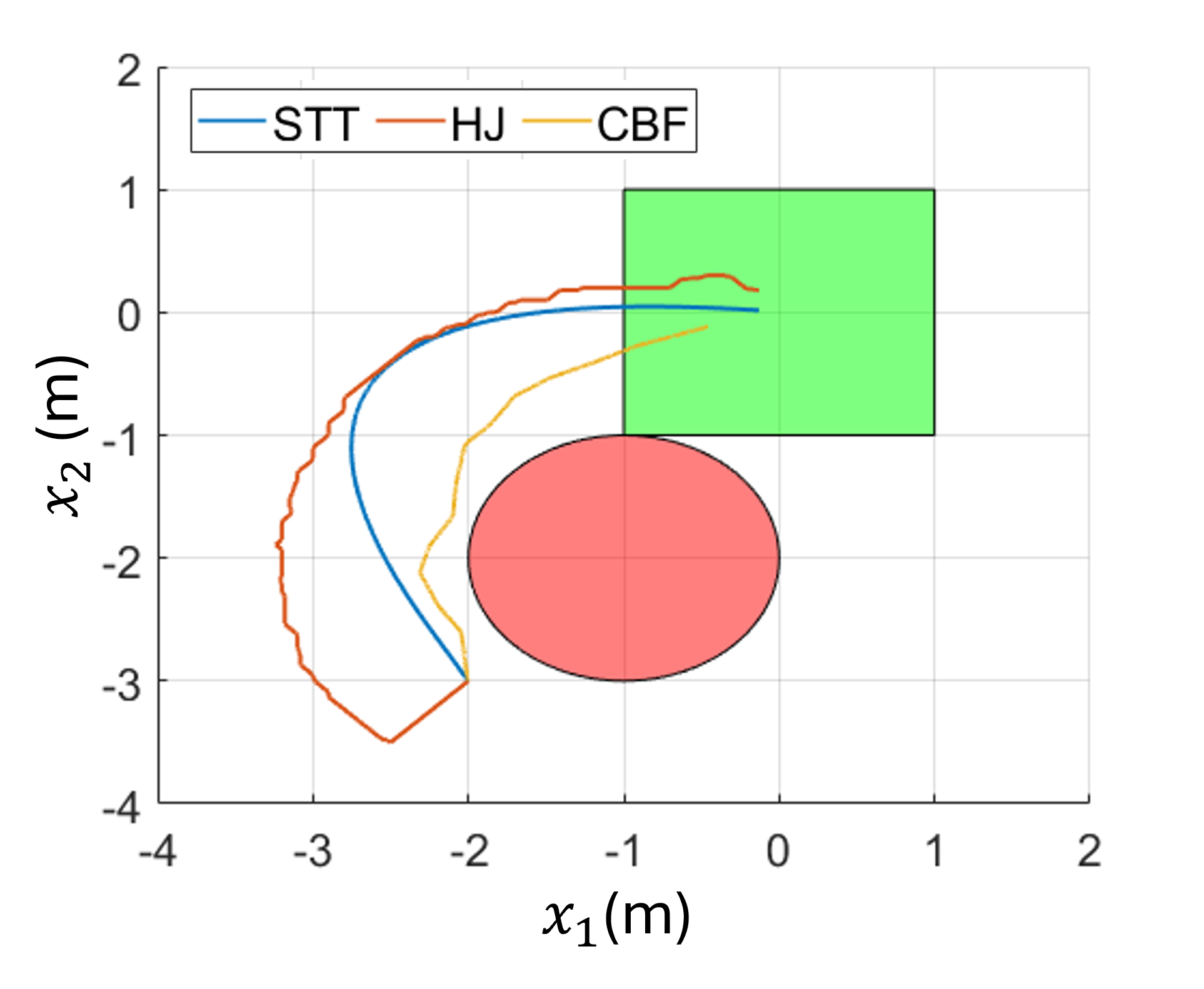} & 
    \includegraphics[width=0.45\textwidth]{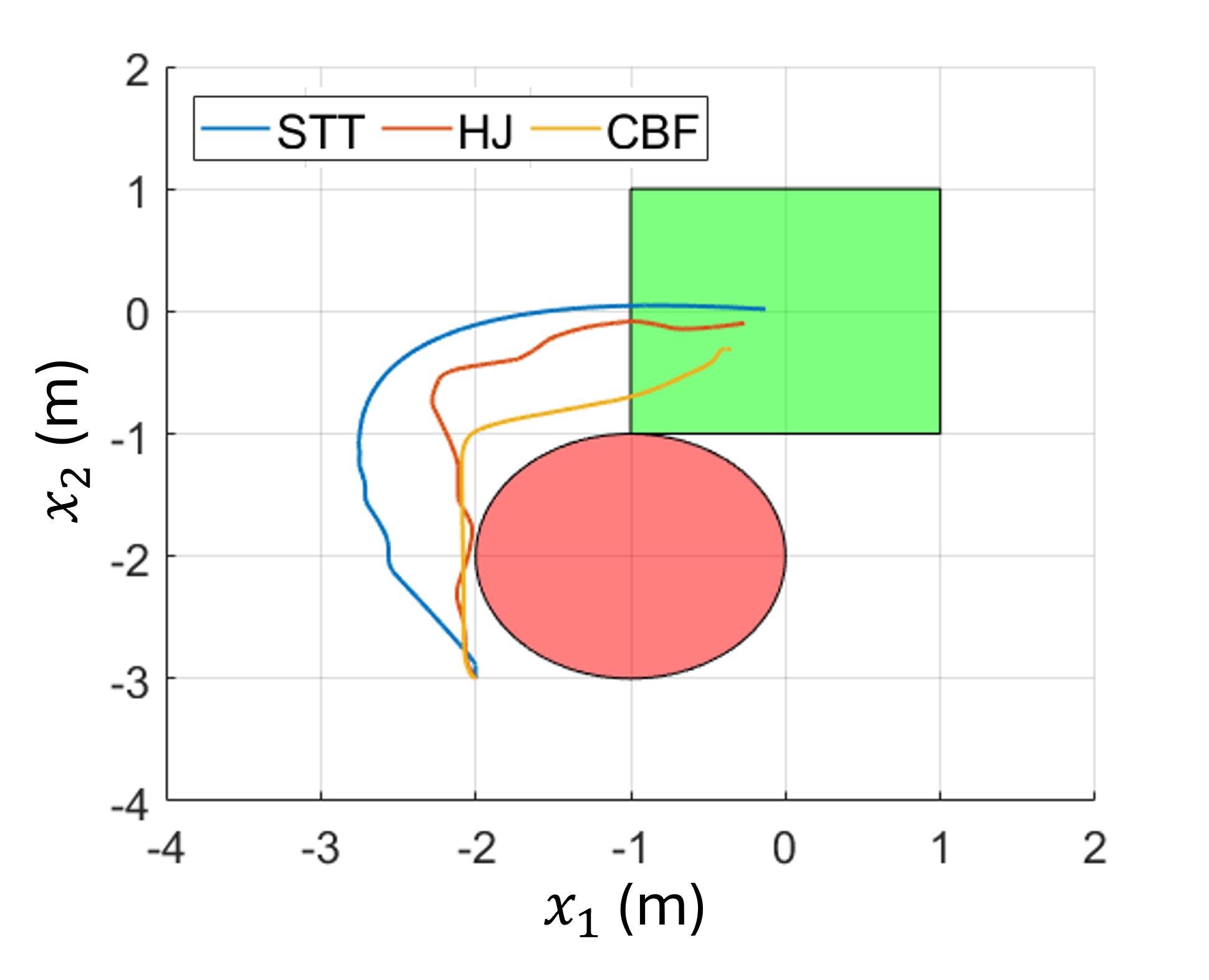} \\
    (a) Single Integrator & (b) Double Integrator \\
    \includegraphics[width=0.45\textwidth]{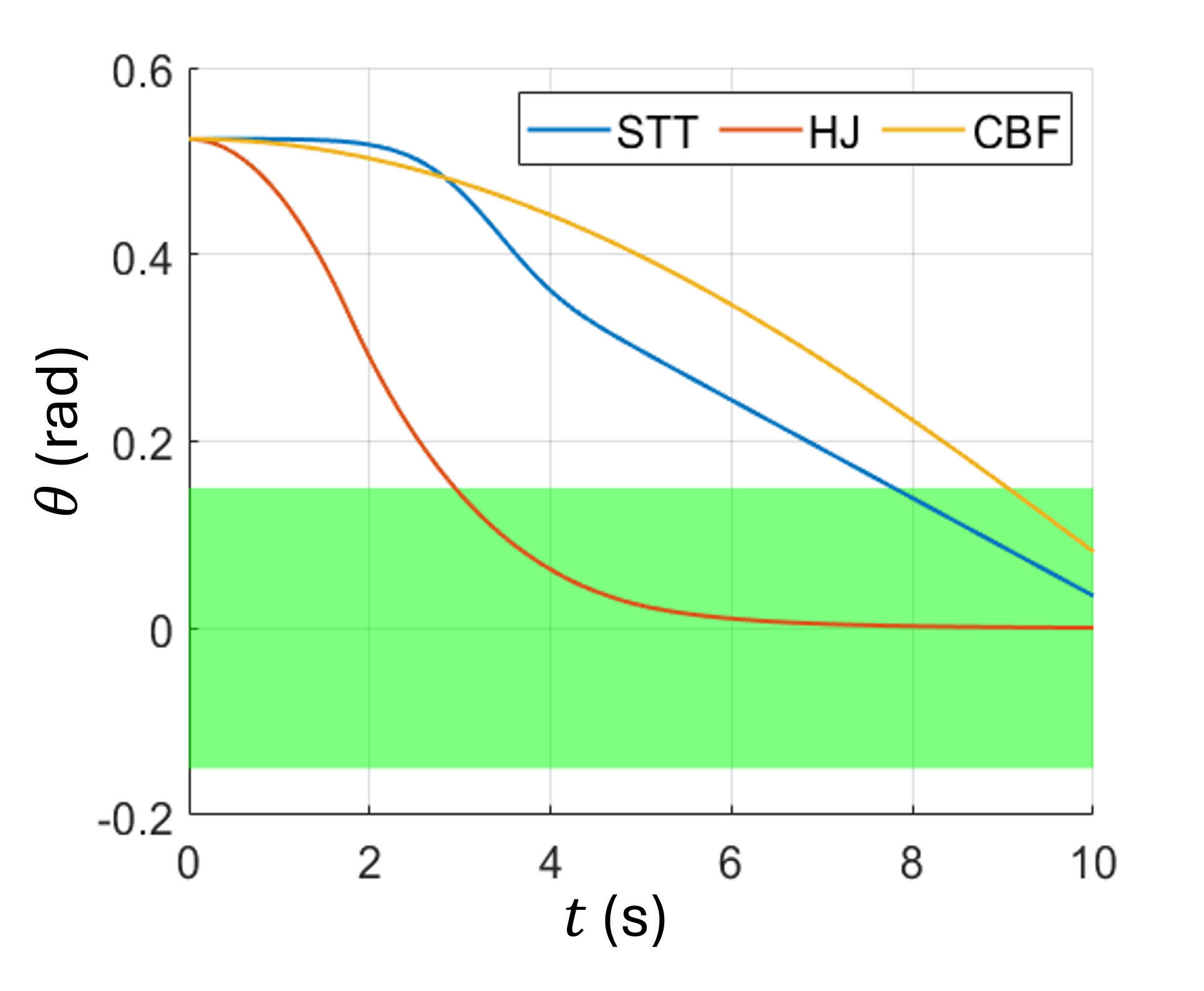} & 
    \includegraphics[width=0.45\textwidth]{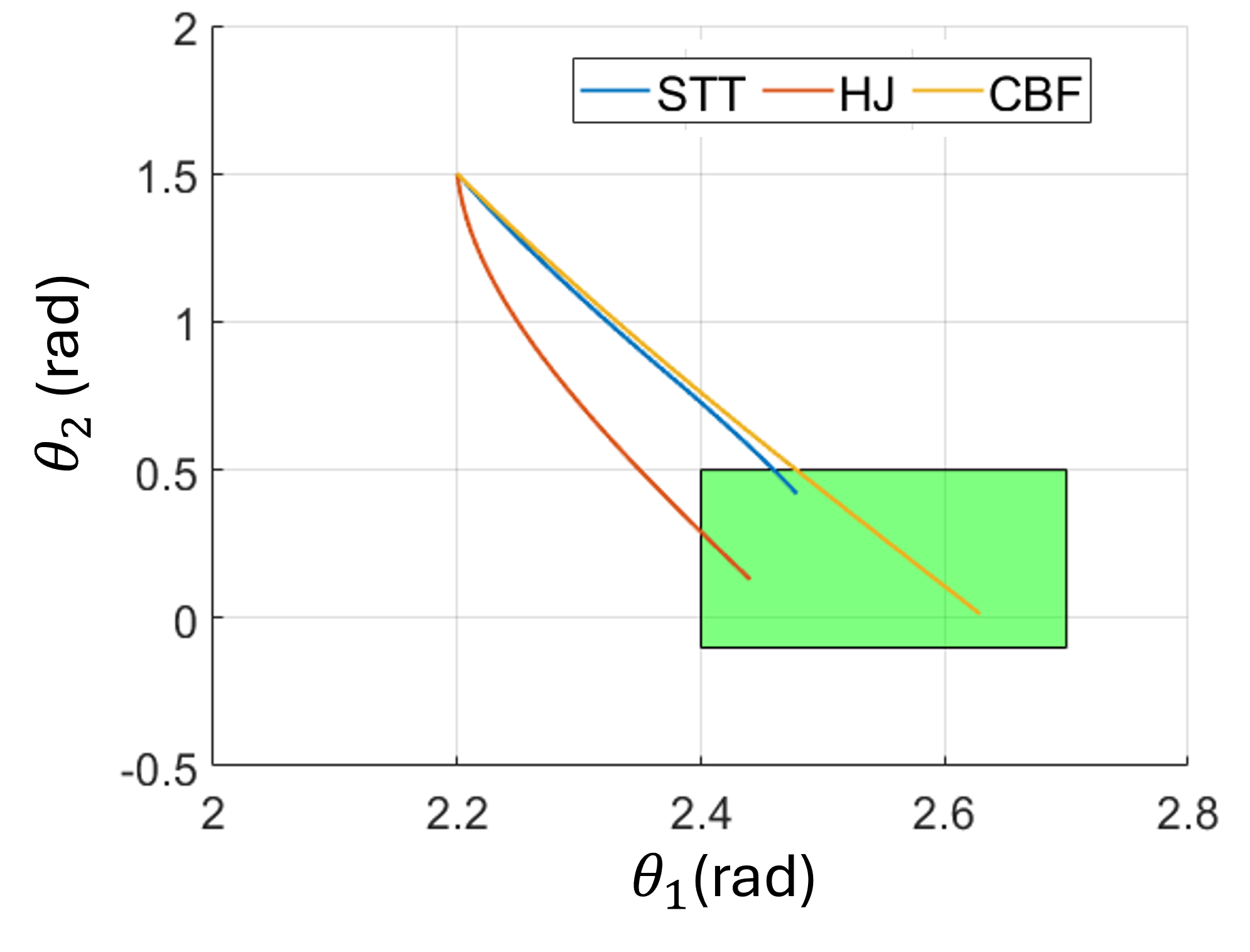} \\
    (c) 1R Manipulator & (d) 2R Manipulator \\
\end{tabular}%
}
\caption{Comparison of STT with CBF and HJB across different systems.}
\label{fig:compare}
\end{figure}

\subsection*{Discussion and Comparison}
Path planning algorithms \cite{qin2023review} offer a potential solution to the case studies, albeit with moderate to high computational complexity and lacking formal guarantees of solution. Conversely, symbolic control techniques \cite{tabuda,SCOTS} promise formal guarantees but come at the expense of increased computational complexity and demanding the knowledge of exact mathematical models. Although it takes moderate time to generate a STT, the sampling-based approach gives a formal guarantee under unknown dynamics to satisfy the T-RAS task. The qualitative comparison of different aspects of the proposed STTs with state-of-the-art algorithms is shown in Table \ref{tab:comp}. Note that it is difficult to handle time-dependent obstacles using state-of-art algorithms. In contrast, STTs can handle them, as shown in the case of the magnetic levitator system.

The efficacy of our approach is also demonstrated across various systems and settings. The omnidirectional mobile robot and 2R manipulator scenarios showcase safe navigation in cluttered two-dimensional workspaces populated with multiple, irregularly shaped unsafe regions. The magnetic levitation system further illustrates the ability of our framework to handle temporally varying unsafe set. Additionally, hardware experiments \href{https://www.youtube.com/watch?v=tigGjCNiwXQ}{Videos} validate the real-time applicability and practical viability of the proposed STT-based control law.
It is also important to note that, the generation of STTs occurs offline, with only the closed-form control law, which is highly computationally efficient, applied in real-time. This design choice significantly enhances scalability, making it suitable for high-dimensional systems.

For quantitative benchmarking, we compared our approach against Hamilton-Jacobi (HJ) \cite{HJ_review} using examples provided in the helperOC toolbox \cite{helperOC} with the help of ToolboxLS \cite{mitchell2007toolbox} and Control Barrier Function (CBF) \cite{CBF} techniques. These include: (i) a 2D single integrator system and (ii) a 2D double integrator system in a reach-avoid task; (iii) a 1R manipulator and (iv) a 2R manipulator, both with reach tasks. Figure \ref{fig:compare} illustrates the simulation results.
The computation times across these approaches are summarized in Table~\ref{tab:comp}. Our method exhibits consistent and significantly lower computational costs compared to HJ and CBF, particularly in higher-dimensional systems. It is important to note that both HJ and CBF methods require knowledge of system dynamics. While learning-based techniques—such as Gaussian Processes (GPs) or Neural Networks (NNs)—could be used to model unknown dynamics for these methods, doing so introduces computational overhead. In contrast, our framework is inherently model-free, scalable, and comes with formal correctness guarantees for satisfying T-RAS tasks under bounded disturbances. 

All the computations were performed on a machine with a Linux Ubuntu operating system with Intel i7-7700 CPU and 32GB RAM. 

\section{Conclusion and Future Work}

The paper focuses on constructing the STTs using a sampling-based approach from the available data of temporal unsafe sets, aiming to ensure the temporal reach-avoid-stay (T-RAS) task. We formulate an SOP to assemble obstacle data and model the tube so that it ensures the T-RAS task is completed with a certified confidence of 1. Consequently, we can design a closed-form, approximation-free control law to retain the unknown MIMO pure-feedback system's trajectory within the tube. We showcase the success of our approach through three different case studies. 

Although the case studies significantly support the proposed approach, the time complexity of constructing the tube for any general case is subject to further investigation as it depends on the length of time horizon, the volume of unsafe region, and the degree of the tube polynomial, hence being tedious to calculate. The proposed approach cannot consider arbitrary input constraints; in future work, we also aim to develop solutions accommodating input constraints. Finally in this work, we focused on fully actuated systems, and in future work, we plan to extend our approach to general underactuated nonlinear systems.


\bibliographystyle{unsrt} 
\bibliography{sources} 

\subsection*{Appendix A. Computation of Lipschitz constants $\mathcal{L}_L$ and $\mathcal{L}_U$}
Employing the results of \cite{Lipschitz2}, we propose the Algorithm \ref{algo:Lipschitz} to estimate the Lipschitz constants of the STTs using a finite number of data collected from the tube boundary. Though we introduce the algorithm for the computation of $\mathcal{L}_L$, one can leverage a similar algorithm to estimate $\mathcal{L}_U$, following the same procedure.
\begin{algorithm}
    \caption{Estimation of $\mathcal{L}_L$ using data}
    \label{algo:Lipschitz}
    \begin{algorithmic}[1]
        \State Select two time instances randomly $t_j$, $t_k$ such that $\lVert t_j - t_k\rVert \leq \alpha, \forall j,k \in [1;\overline{N}], \alpha \in \R_{>0}$ 
        \State Calculate $\theta_j^i = \gamma_{i,L}(c_{i,L},t_j) $ and $ \theta_k^i = \gamma_{i,L}(c_{i,L},t_k)$ where $\theta_j^i$ and $\theta_k^i$ denotes the lower bound of the tube in $i$-th dimension at $j$-th and $k$-th time instances 
        \State Compute slope $S_{jk}^i = \frac{\lVert \theta_j^i - \theta_k^i \rVert}{t_j - t_k} , \quad \forall j \neq k$ 
        \State Compute maximum slope as $\Psi_i = \max \{S_{jk}^i|j,k \in [1;\overline{N}], j \neq k\}$ 
        \State Repeat steps 1-4 $M$ times and obtain $\Psi_{i,1}, \ldots, \Psi_{i,M}$ 
        \State Apply Reverse Weibull Distribution \cite{Lipschitz2} to $\Psi_{i,1}, \ldots, \Psi_{i,M}$, which gives us so-called location, scale and shape parameters
        \State The obtained \textit{location parameter} is denoted by $\mathcal{L}_i$
        \State Repeat steps 1-7 to get $\mathcal{L}_1, \ldots, \mathcal{L}_n$ 
        \State $\mathcal{L}_L = \max \{\mathcal{L}_1, \ldots, \mathcal{L}_n\}$
    \end{algorithmic}
\end{algorithm}

 The following Lemma \ref{lem:Lipschitz}, borrowed from \cite{Lipschitz2}, under the proposed algorithm ensures the convergence of the estimated Lipschitz constants to their actual values.

\begin{lemma}[\cite{Lipschitz}]\label{lem:Lipschitz}
    In the verge of Algorithm \ref{algo:Lipschitz}, the estimated Lipschitz constants $\mathcal{L}_L$ and $\mathcal{L}_U$ tends to their actual values if and only if $\alpha$ goes to zero with $\overline{N}, M$ tends to infinity.
\end{lemma}
\quad Note that, picking minimal value of $\alpha$ and very high value of $\overline{N}, M$ will give a precise approximation of the Lipschitz constants.

\end{document}